\documentclass[11pt,reqno]{amsart}
\usepackage[top=1.2in, bottom=1.2in, left=1.2in, right=1.2in]{geometry}
\usepackage{amsfonts, amssymb, amsmath, mathtools, dsfont, xcolor}

\usepackage[linktoc=page, colorlinks, linkcolor=blue, citecolor=blue]{hyperref}
\usepackage{enumerate, commath,bm}

\usepackage{graphicx}

\usepackage{subcaption}
\usepackage[font=small]{caption}
\numberwithin{equation}{section}

\DeclareMathOperator{\E}{\mathbb{E}}

\DeclareMathOperator{\supp}{supp}

\DeclareMathOperator{\Lap}{Lap}
\DeclareMathOperator{\dens}{dens}
\DeclareMathOperator{\diam}{diam}
\DeclareMathOperator{\TSP}{TSP}
\DeclareMathOperator{\MST}{MST}

\DeclareMathOperator{\range}{range}

\renewcommand{\Pr}[2][]{\mathbb{P}_{#1} \left\{ #2 \rule{0mm}{3mm}\right\}}

\newcommand{\ip}[2]{\langle#1,#2\rangle}

\def \N {\mathbb{N}}

\def \R {\mathbb{R}}

\def \Z {\mathbb{Z}}

\def \AA {\mathcal{A}}

\def \FF {\mathcal{F}}

\def \MM {\mathcal{M}}
\def \NN {\mathcal{N}}

\def \BB {\mathcal{B}}
\def \XX {\mathcal{X}}

\def \a {\alpha}

\def \e {\varepsilon}
\def \d {\delta}
\def \l {\lambda}

\def \tran {\mathsf{T}}

\def \TV {\mathrm{TV}}

\def \Lip {\mathrm{Lip}}
\def \one {{\textbf 1}}

\def \Npack {{N_\textrm{pack}}}

\newtheorem{theorem}{Theorem}[section]
\newtheorem{proposition}[theorem]{Proposition}
\newtheorem{corollary}[theorem]{Corollary}
\newtheorem{lemma}[theorem]{Lemma}

\newtheorem{definition}[theorem]{Definition}

\theoremstyle{remark}
\newtheorem{remark}[theorem]{Remark}

\begin{document}

\title{Private Measures, Random Walks, and Synthetic Data}

\author{March Boedihardjo}
\address{Department of Mathematics, Michigan State University}
\email{boedihar@msu.edu}
\author{Thomas Strohmer}
\address{Department of Mathematics, University of California Davis}
\email{strohmer@math.ucdavis.edu}
\author{Roman Vershynin}
\address{Department of Mathematics, University of California Irvine}
\email{rvershyn@uci.edu}

\begin{abstract}
Differential privacy is a mathematical concept that 
provides an information-theoretic security guarantee.
While differential privacy has emerged as a de facto standard for guaranteeing privacy in data sharing, the known mechanisms to achieve it come with some serious limitations. 
Utility guarantees are usually provided only for a fixed, a priori specified set of queries.
Moreover, there are no utility guarantees for more complex---but very common---machine learning tasks such as clustering or classification. In this paper we overcome some of these limitations. Working with metric privacy, a powerful generalization of differential privacy, we develop a polynomial-time algorithm that creates a private measure from a data set. This private measure allows us to efficiently construct private synthetic data that are accurate for a wide range of  statistical analysis tools. Moreover, we prove an asymptotically sharp min-max result for private measures  and synthetic data in general compact metric spaces, for any fixed privacy budget  bounded away from zero.
A key ingredient in our construction is a new superregular random walk, whose joint distribution of steps is as regular as that of independent random variables, yet which deviates from the origin logarithmically slowly.
\end{abstract}



\maketitle

\section{Introduction} \label{s:intro}

\subsection{Motivation}

The right to privacy is enshrined in the Universal Declaration of Human Rights~\cite{assembly1948universal}. However, as artificial intelligence is  more and more permeating our daily lives, data sharing is increasingly locking horns with data privacy concerns.
Differential privacy (DP), a probabilistic mechanism that provides an information-theoretic privacy guarantee, has emerged as a de facto standard for implementing privacy in data sharing~\cite{dwork2014algorithmic}.  For instance, DP has been adopted by several tech companies~\cite{dwork2019differential} and will also be used in connection with the release of the Census 2020 data~\cite{abowd2018us,abowd2019census}. 

Yet, current embodiments of DP come with some serious limitations~\cite{domingo2021limits,hauer2021differential,wezerek}: 
\begin{enumerate}[(i)]
\item  Utility guarantees are usually provided only for a fixed set of queries. This means that either DP has to be used in an interactive scenario or the queries have to specified in advance. 
\item There are no utility guarantees for more complex---but very common---machine learning tasks such as clustering or classification. 
\item DP can suffer from a poor privacy-utility tradeoff, leading to either insufficient privacy protection or to data sets of rather low utility, thereby making DP of limited use in many applications~\cite{domingo2021limits}.
\end{enumerate}

Another approach to enable privacy in data sharing is based on the concept of synthetic data~\cite{bellovin2019privacy}. The goal of synthetic data is  to create a dataset that maintains the statistical properties of the original data while not exposing sensitive information. The combination of differential privacy with synthetic data has been suggested as a  best-of-both-world solutions~\cite{hardt2012simple,bellovin2019privacy,kearnsroth2020,liu2021leveraging,BSV2021a}.
While combining DP with synthetic data can indeed provide more flexibility and thereby partially address some of the issues in~(i), in and of itself it is not a panacea for the aforementioned problems. 

One possibility to construct  differentially private synthetic datasets that are not tailored to a priori specified queries is to simply add independent Laplacian noise to each data point. However, the amount noise that has to be added to achieve sufficient DP is too large with respect to maintaining satisfactory utility even for basic counting queries~\cite{xiao2010differential}, not to mention more sophisticated machine learning tasks.

This raises the fundamental question whether it is even possible to construct in a numerically efficient manner differentially private synthetic data that come with rigorous utility guarantees for a wide range of (possibly complex) queries, while achieving a favorable privacy-utility tradeoff? In this paper we will answer this question to the affirmative.

\subsection{A private measure}

A main objective of this paper is to construct a {\em private measure} on a given metric space $(T,\rho)$. Namely, we design an algorithm that transforms a probability measure $\mu$ on $T$ into another probability measure $\nu$ on $T$, and such that this transformation is both private and accurate.

For clarity, let us first consider the special case of empirical measures, where our goal can be understood as creating {\em differentially private synthetic data}. Specifically, we are looking for a computationally tractable algorithm that transforms true input data $X =(X_1,\ldots,X_n) \in T^n$ into synthetic output data $Y=(Y_1,\ldots,Y_m) \in T^m$ for some $m$, and 
which is $\e$-differentially private (see Definition~\ref{def: DP}) and such that
the empirical measures 
$$
\mu_X = \frac{1}{n} \sum_{i=1}^n \d_{X_i}
\quad \text{and} \quad 
\mu_Y = \frac{1}{m} \sum_{i=1}^m \d_{Y_i}
$$
are close to each other in the Wasserstein 1-metric (recalled in Section~\ref{s: Wasserstein}): 
\begin{equation}	\label{eq: accuracy goal}
\E W_1 \left( \mu_X,\mu_Y \right) \le \gamma,
\end{equation}
where $\gamma>0$ is as small as possible.
In other words, our goal is to create synthetic data $Y$ from the true data $X$ by adding noise of average magnitude $\gamma$, just not necessarily i.i.d. noise.

The main result of this paper is a computationally effective private algorithm 
whose accuracy $\gamma$ that is expressed in terms of the multiscale geometry of the metric space $(T,\rho)$. A consequence of this result, 
Theorem~\ref{thm: asymp synth data}, states that 
if the metric space has Minkowski dimension $d \ge 1$, then, 
ignoring the dependence on $\e$ and lower-order terms in the exponent, 
we have
\begin{equation}	\label{eq: target accuracy}
\E W_1 \left( \mu_X,\mu_Y \right) \sim n^{-1/d}
\end{equation}
The dependence on $n$ is optimal and quite intuitive. Indeed, if the
true data $X$ consists of $n$ i.i.d. random points chosen uniformly from the unit cube $T=[0,1]^d$, then the average spacing between these points is of the order $n^{-1/d}$.
So our result shows that privacy can be achieved
by a {\em microscopic perturbation}, one whose magnitude is 
roughly the same as the average spacing between the points.

Our more general result, Theorem~\ref{thm: private metric}, holds 
for arbitrary compact metric spaces $(T,\rho)$
and, more importantly, for general input measures (not just empirical ones).
To be able to work in such generality, 
we employ the notion of {\em metric privacy} which 
reduces to differential privacy when we specialize to empirical measures
(Section~\ref{s: metric privacy}).

\subsection{Uniform accuracy over Lipschitz statistics}

 The choice of the Wasserstein 1-metric to quantify accuracy ensures that all Lipschitz statistics are preserved uniformly. 
Indeed, by the Kantorovich-Rubinstein duality theorem, \eqref{eq: accuracy goal} yields
\begin{equation} \label{eq:lip1}
\E \sup_f \Big\vert \frac{1}{n} \sum_{i=1}^n f(X_i) - \frac{1}{m} \sum_{i=1}^m f(Y_i) \Big\vert
\le \gamma
\end{equation}
where the supremum is over all $1$-Lipschitz functions $f:T \to \R$.

Standard private synthetic data generation methods that come with rigorous accuracy guarantees  do so
with respect to a {\em predefined} set of linear queries, such as low-dimensional marginals, see e.g.~\cite{barak2007privacy,thaler2012faster,dwork2015efficient,BSV2021a}. While this may suffice in some cases, there is no assurance that the synthetic data behave in the same way as the original data under more complex, but frequently employed, machine learning techniques. For instance, if we want to apply a clustering method to the synthetic data, we cannot be sure that the results we get are close to those for the true data. This can drastically limit effective and reliable analysis of synthetic data.

In contrast, since the synthetic data constructed via our proposed method satisfy a {\em uniform} bound~\eqref{eq:lip1}, this  provides data analysts with a vastly increased toolbox of machine learning methods for which one can expect outcomes that are similar for the original data 
and the synthetic data.

As concrete examples let us look at two of the most common tasks in machine learning, namely clustering and classification. 
While not every clustering method will satisfy a Lipschitz property, there do exist Lipschitz clustering functions that achieve state-of-the-art results, see e.g.~\cite{kovalev2021lipschitz,yang2020clustering}. Similarly,  there is distinct interest in Lipschitz function based classifiers, since they are more robust and less susceptible to adversarial attacks. This includes conventional classification methods such as support vector machines~\cite{von2004distance} as well as classifiers based on Lipschitz neural networks~\cite{virmaux2018lipschitz,bethune2021many}.  These are just a few examples of complex machine learning tools that can be reliably applied to the synthetic data constructed via our private measure algorithm. Moreover, since our results hold for general compact metric spaces, this paves the way for creating private synthetic data for a wide range of data types. We will present a detailed algorithmic and numerical investigation of the proposed method in a forthcoming paper.

\subsection{A superregular random walk}

The most popular way to achieve privacy is by adding random noise, typically either by adding an appropriate amount of Laplacian noise or Gaussian noise (these methods are aptly referred to as {\em Laplacian mechanism} and {\em Gaussian mechanism}, respectively~\cite{dwork2014algorithmic}). We, too, can try to make a probability measure $\mu$ on $T$ private by discretizing $T$ (replacing it with a finite set of points) and then adding random noise to the weights of the points. Going this route, however, yields suboptimal results. For example, it is not difficult to check that if $T$ is the interval $[0,1]$, the accuracy of the Laplacian mechanism cannot be better than $n^{-1/2}$, which is suboptimal compared to optimal accuracy $n^{-1}$ in \eqref{eq: target accuracy}.  

This loss of accuracy is caused by the accumulation of additive noise. Indeed, adding $n$ independent random variables of unit variance produces noise of the order $n^{1/2}$. This prompts a basic probabilistic question: can we construct $n$ random variables that are ``close'' to being independent, but whose partial sums cancel more perfectly than those of independent random variables? We answer this question affirmatively in Theorem~\ref{thm: bounded random walk}, where we construct random variables $Z_1,\ldots,Z_n$ whose joint distribution is as regular as that of i.i.d.\ Laplacian random variables, yet whose partial sums grow {\em logarithmically} as opposed to $n^{1/2}$:
$$
\max_{1 \le k \le n} \E \big( Z_1+\cdots+Z_k \big)^2 = O(\log^3 n).
$$
One can think of this as a random walk that is locally similar to the one with i.i.d.\ steps, but is globally much more bounded. Our construction is a nontrivial modification of L\'evy's construction of Brownian motion. It may be interesting and useful beyond applications to privacy.

\subsection{Comparison to existing work}

The numerically efficient construction of accurate differentially private synthetic data is highly non-trivial. As case in point, Ullman and Vadhan~\cite{ullman2011pcps} showed (under standard cryptographic assumptions) that in general it is NP-hard to make private synthetic Boolean data which approximately preserve all two-dimensional marginals. 
There exists a substantial body of work for generating privacy-preserving synthetic data, cf.~e.g.~\cite{abowd2001disclosure,burridge2003information,abay2018privacy,dahmen2019synsys,mendelevitch2021fidelity}, but---unlike our work---without providing any rigorous privacy or accuracy guarantees. Those papers on synthetic data that do provide rigorous guarantees are limited to accuracy bounds for a finite set of a priori specified 
queries, see for example~\cite{barak2007privacy,blum2013learning,thaler2012faster,dwork2015efficient,BSV2021a,BSV2021c}, see also the tutorial~\cite{vadhan2017complexity}. As discussed before, this may suffice for specific purposes, but in general severely limits the impact and usefulness of synthetic data. In contrast, the present work provides accuracy guarantees for a wide range of machine learning techniques. Furthermore, our our results hold for general compact metric spaces, as we establish metric privacy instead of just differential privacy.

A special example of the topic investigated in this paper is the publication of differentially private histograms, which is a well studied problem in the privacy literature, see e.g.~\cite{hay2009boosting,qardaji2013understanding,meng2017different,xu2013differentially,xiao2010differential,nelson2019chasing,zhang2016privtree,abowd2019census} and Chapter 4 in~\cite{li2016differential}. 
In the specific context of histograms, the Haar function based approach  to construct a superregular random walk proposed in our paper is related to the wavelet-based method~\cite{xiao2010differential} and to other hierarchical histogram partitioning methods~\cite{hay2009boosting,qardaji2013understanding,zhang2016privtree}. Like our approach, \cite{hay2009boosting,xiao2010differential} obtain consistency of counting queries across the hierarchical levels, owing to the specific way that noise is added.
Also, the accuracy bounds obtained in~\cite{hay2009boosting,xiao2010differential} are similar to ours, as they are also polylogarithmic (although we are able to obtain a smaller exponent).  There are, however, several {\em key differences}. While our approach gives a convenient way to generate accurate and differentially private {\em synthetic data} $Y$ from true data $X$, the methods of the aforementioned papers are not suited to create synthetic data. Instead, these methods release answers to queries. Moreover, accuracy is proven for just {\em a single} given range query and not simultaneously for all queries like we do. This limitation makes it impossible to create accurate synthetic data with the algorithms in~\cite{hay2009boosting,xiao2010differential}. Moreover, unlike the aforementioned papers, our work allows the data to be quite general, since we prove metric privacy and not just differential privacy. Furthermore, our results apply to multi-dimensional data, and are not limited to the one-dimensional setting. 

Perhaps closest to our paper is~\cite{wang2016differentially}, where the authors consider differentially private synthetic data in $[0, 1]^d$ with guarantees for any smooth queries with bounded partial derivatives of order $K$. The case $K=1$  corresponds to 1-Lipschitz functions considered in our paper. In that case~\cite{wang2016differentially} obtains an accuracy 

There exist several papers on the private estimation of  density and other statistical quantities~\cite{kamath2019privately,duchi2018minimax}, and sampling from distributions in a private manner is the topic of~\cite{raskhodnikova2021differentially}. While definitely interesting, that line of work is not concerned with synthetic data, and thus there is little overlap with this work.

\subsection{The architecture of the paper}

The remainder of this paper is organized as follows. We introduce some background material and notation in Section~\ref{s:notation}, such as the concept of {\em metric privacy} which generalizes differential privacy.  In Section~\ref{s: random walk} we construct a {\em superregular random walk} (Theorem~\ref{thm: bounded random walk}). We analyze metric privacy in more detail in Section~\ref{s: MP}, where we also provide a link from the general private measure problem to private synthetic data (Lemma~\ref{lem: PM to SD}). In Section~\ref{s:pm_line}
we use the superregular random walk to construct a {\em private measure on the interval} $[0,1]$ (Theorem~\ref{thm: pm on interval}). In Section~\ref{s: width} we use a link between the Traveling Salesman Problem and minimum spanning trees to devise a {\em folding technique}, which we apply in Section~\ref{s:pm_mp} to ``fold''  the interval into a space-filling curve to construct a {\em private measure on a general metric space} (Theorem~\ref{thm: private metric}). Postprocessing the private measure with quantization and splitting, we then generate {\em private synthetic data} in a general metric space (Corollary~\ref{cor: private synthetic data}). In Section~\ref{s:lower} we turn to {\em lower bounds} for private measures (Theorem~\ref{thm: private measure lower}) and synthetic data (Theorem~\ref{thm: synthetic data lower}) on a general metric space.
We do this by employing a technique of Hardt and Talwar, which we present in a Proposition~\ref{prop: synthetic lower} that identifies general limitations for synthetic data. 
In Section~\ref{s:examples} we illustrate our general results on a specific example of a metric space: the Boolean cube $[0,1]^d$. We construct a private measure (Corollary~\ref{cor: PM cube}) and private synthetic data (Corollary~\ref{cor: private synthetic cube}) on the cube, and show near optimality of these results in Corollary~\ref{cor: PM cube lower} and Corollary~\ref{cor: synth data cube lower}, respectively. Results similar to the ones for the $d$-dimensional cube hold for arbitrary metric space of Minkowski dimension $d$. For any such space, we prove {\em asymptotically sharp min-max results} for private measures (Theorem~\ref{thm: asymp PM}) and synthetic data (Theorem~\ref{thm: asymp synth data}).

\section{Background and Notation}\label{s:notation}

The motivation behind the concept of differential privacy is the desire to  
protect an individual's data, while publishing aggregate information about the database~\cite{dwork2014algorithmic}. Adding or removing  the data of one individual should have a negligible effect on the query outcome, as formalized in the following definition.

\begin{definition}[Differential Privacy~\cite{dwork2014algorithmic}]  \label{def: DP}
  A randomized algorithm ${\mathcal M}$ gives $\e$-differential privacy
 if for any input databases $D$ and $D'$ differing on at most one element, 
 and any measurable subset $S \subseteq \range({\mathcal M})$, we have 
 $$
 \frac{\Pr{{\mathcal M}(D) \in S}}{\Pr{{\mathcal M}(D') \in S}} \le \exp(\e),
 $$
where the probability is with respect to the randomness of  ${\mathcal M}$.
\end{definition}

\subsection{Defining metric privacy} 		\label{s: metric privacy}

While differential privacy is a concept of the discrete world (where datasets can differ in a {\em single element}), it is often desirable to have more freedom in the choice of input data.
The following general notion (which seems to be known under slightly different, and somewhat less general, versions, see e.g.~\cite{ABCP} and the references therein) extends the classical concept of differential privacy.

\begin{definition}[Metric privacy]		\label{def: metric privacy}
  Let $(T,\rho)$ be a compact metric space and $E$ be a measurable space.
  A randomized algorithm $\AA: T \to E$  
  is called $\a$-metrically private 
  if, for any inputs $x,x' \in T$ and any measurable subset $S \subset E$, we have
  \begin{equation}	\label{eq: metric DP}
  \frac{\Pr{\AA(x) \in S}}{\Pr{\AA(x') \in S}} \le \exp \left( \a \, \rho(x,x') \right).
  \end{equation}
\end{definition}

To see how this metric privacy encompasses differential privacy,
consider a product space $T = \XX^n$ 
and equip it with the Hamming distance
\begin{equation}	\label{eq: Hamming}
\rho_H(x,x') = \abs{\{i \in [n]:\; x_i \ne x'_i\}}.
\end{equation}
The $\a$-differentially privacy of an algorithm $\AA : \XX^n \to E$ can be expressed as 
\begin{equation}	\label{eq: classical DP}
\frac{\Pr{\AA(x) \in S}}{\Pr{\AA(x') \in S}} \le \exp \left( \a \right)
\quad \text{whenever } \rho_H(x,x') \le 1.
\end{equation}
Note that \eqref{eq: classical DP} is equivalent to \eqref{eq: metric DP} for $\rho=\rho_H$. Obviously, \eqref{eq: metric DP} implies \eqref{eq: classical DP}. The converse implication can be proved by replacing one coordinate of $x$ by the corresponding coordinate of $x'$ and applying \eqref{eq: classical DP}
$\rho_H(x,x')$ times, then telescoping. Let us summarize: 

\begin{lemma}[MP vs.\ DP]			\label{lem: MP vs DP}
  Let $\XX$ be an arbitrary set. Then an algorithm $\AA: \XX^n \to E$
  is $\alpha$-differentially private if an only if $\AA$ is $\alpha$-metrically private 
  with respect to the Hamming distance \eqref{eq: Hamming} on $\XX^n$.
\end{lemma}

Unlike differential privacy, metric privacy goes beyond product spaces, and thus allows the data to be quite general. In this paper, for example, the input data are probability measures. Moreover, metric privacy does away with the assumption that the data sets $D, D'$ be different in a {\em single element}. This assumption is sometimes too restrictive: general measures, for example, do not break down into natural single elements.

\subsection{Distances between measures}

This paper will use three classical notions of distance between measures.

\subsubsection{Total variation}		\label{s: TV}
The total variation (TV) norm~\cite[Section III.1]{dunford1958partI} of a signed measure $\mu$ on a measurable space $(\Omega,\FF)$ is defined as\footnote{The factor $2$ is chosen for convenience.}
\begin{equation}	\label{eq: TV}
\norm{\mu}_\TV = \frac{1}{2} \sup_{\Omega = \cup_i A_i} \sum_i \abs{\mu(A_i)} 
\end{equation}
where the supremum is over all partitions $\Omega$ into countably many parts $A_i \in \FF$.
If $\Omega$ is countable, we have
\begin{equation}	\label{eq: TV as L1}
\norm{\mu}_\TV = \frac{1}{2} \sum_{\omega \in \Omega} \abs{\mu(\{\omega\})}. 
\end{equation}
The TV distance between two probability measures $\mu$ and $\nu$ is defined 
as the TV norm of the signed measure $\mu-\nu$. Equivalently,
$$
\norm{\mu-\nu}_\TV 
= \sup_{A \in \FF} \abs{\mu(A)-\nu(A)}.
$$

\subsubsection{Wasserstein distance}			\label{s: Wasserstein}
Let $(\Omega,\rho)$ be a bounded metric space. We define the Wasserstein 1-distance (henceforth simply referred to as Wasserstein distance) between probability measures $\mu$ and $\nu$ on $\Omega$  as~\cite{villani2009optimal}
\begin{equation}\label{def:W1}
W_1(\mu,\nu) = \inf_\gamma \int_{\Omega \times \Omega} \rho(x,y) \, d\gamma(x,y)
\end{equation}
where the infimum is over all {\em couplings} $\gamma$ of $\mu$ and $\nu$, or 
probability measures on $\Omega \times \Omega$ whose marginals on the first and second coordinates are $\mu$ and $\nu$, respectively. In other words, $W_1(\nu,\mu)$ minimizes the {\em transportation cost} between the ``piles of earth'' $\mu$ and $\nu$. 

The Kantorovich-Rubinstein duality theorem~\cite{villani2009optimal} gives an equivalent representation:
$$
W_1(\mu,\nu) = \sup_{\norm{h}_\Lip \le 1} \left( \int h\, d\mu - \int h\, d\nu \right)
$$
where the supremum is over all continuous, $1$-Lipschitz functions $h: \Omega \to \R$.

For probability measures $\mu$ and $\nu$ on $\R$,
the Wasserstein distance has the following representation, according to Vallender \cite{Vallender}:
\begin{equation}	\label{eq: Vallender}
W_1(\mu,\nu) = \norm{F_\mu-F_\nu}_{L^1(\R)}.
\end{equation}
Here $F_\mu(x) = \mu \left( (-\infty,x] \right)$ is the cumulative distribution function of $\mu$, 
and similarly for $F_\nu(x)$.

Vallender's identity \eqref{eq: Vallender} can be used to define Wasserstein distance for 
{\em signed measures} on $\R$. Moreover, for signed measures on $[0,1]$, the Wasserstein distance defined this way is always finite, and it defines a pseudometric.

\section{A superregular random walk}		\label{s: random walk}

The classical random walk with independent steps of unit variance is not bounded: 
it deviates from the origin at the expected rate $\sim n^{1/2}$. 
Surprisingly, there exists a random walk whose joint distribution of steps 
is as regular as that of independent Laplacians, yet that deviates from the origin {\em logarithmically} slowly.

\begin{theorem}[A superregular random walk]	\label{thm: bounded random walk}
  For every $n \in \N$, there exists a probability density of the form $f(z) = \frac{1}{\beta} e^{-V(z)}$ on $\R^n$ that satisfies the following two properties.
  \begin{enumerate}[(i)]
    \item (Regularity): the potential $V$ is $1$-Lipschitz in the $\ell^1$ norm, i.e.
     \begin{equation}\label{eq:regularity}
        \abs{V(x)-V(y)} \le \norm{x-y}_1 \quad \text{for all } x,y \in \R^n.
     \end{equation}
    \item (Boundedness): a random vector $Z=(Z_1,\ldots,Z_n)$ distributed according 
      to the density $f$ satisfies
      \begin{equation}\label{eq:boundedness2}
      \E(Z_1+\cdots+Z_k)^2 \le C \log^3 n
      \quad \text{for all } 1 \le k \le n,
      \end{equation}
where $C>0$ is a universal constant.
  \end{enumerate}
\end{theorem}

\subsection{Heuristics}

We will define a superregular random walk by modifying the L\'evy's construction of the Brownian motion.  In this construction, the path of a Brownian motion on $[0,1]$ is defined as a random Gaussian series with respect to the Faber-Schauder basis of the space of continuous functions, see~\cite[Section IX.1]{bhattacharya2007basic}. We will replace Gaussian weights by Laplacian weights with smaller variances, and truncate the series.

To that end, recall the definition of the {\em Faber-Schauder system} of ``hat functions'' 
$\phi_1,\phi_2,\ldots$ on the interval $[0,1]$. First, we set
$$
\phi_1(t) = t.
$$
Next, for each level $\ell \in \N$ and each $k \in \{1,\ldots,2^{\ell-1}\}$, 
we define $\phi_{2^{\ell-1}+k}(t)$ as the function on $[0,1]$ that takes value $0$ outside the interval 
\begin{equation}	\label{eq: Schauder support}
(a_{\ell,k}, b_{\ell,k}) \coloneqq \left( \frac{k-1}{2^{\ell-1}}, \, \frac{k}{2^{\ell-1}} \right),
\end{equation}
takes value $1$ at the midpoint $c_k \coloneqq (2k-1)/2^\ell$ of the interval, and interpolates linearly in between. The Faber-Schauder system forms a Shcauder basis in the Banach space $C_0[0,1]$ of continuous functions that take zero value at the origin. 

The Faber-Schauder system is conveniently organized by levels $\ell=0,1,2,\ldots$ We place a single function $\phi_1(t)=t$ is on level $\ell=0$, and each subsequent level $\ell \ge 1$ contains $2^{\ell-1}$ functions $\phi_j$ supported on disjoint intervals $(a_{\ell,k}, b_{\ell,k})$ of length $1/2^{\ell-1}$.
Throughout this section, $\ell(j)$ will denote the level the function $\phi_j$ belongs to, 
e.g. $\ell(1)=0$, $\ell(2)=1$, $\ell(3)=\ell(4)=2$, $\ell(5)=\ell(6)=\ell(7)=\ell(8)=3$, etc.
See Figure~\ref{fig: schauder} for an illustration of these functions.

\begin{figure}[htp]	
     \begin{subfigure}{.22\textwidth}
        \centering
        \includegraphics[width=\linewidth]{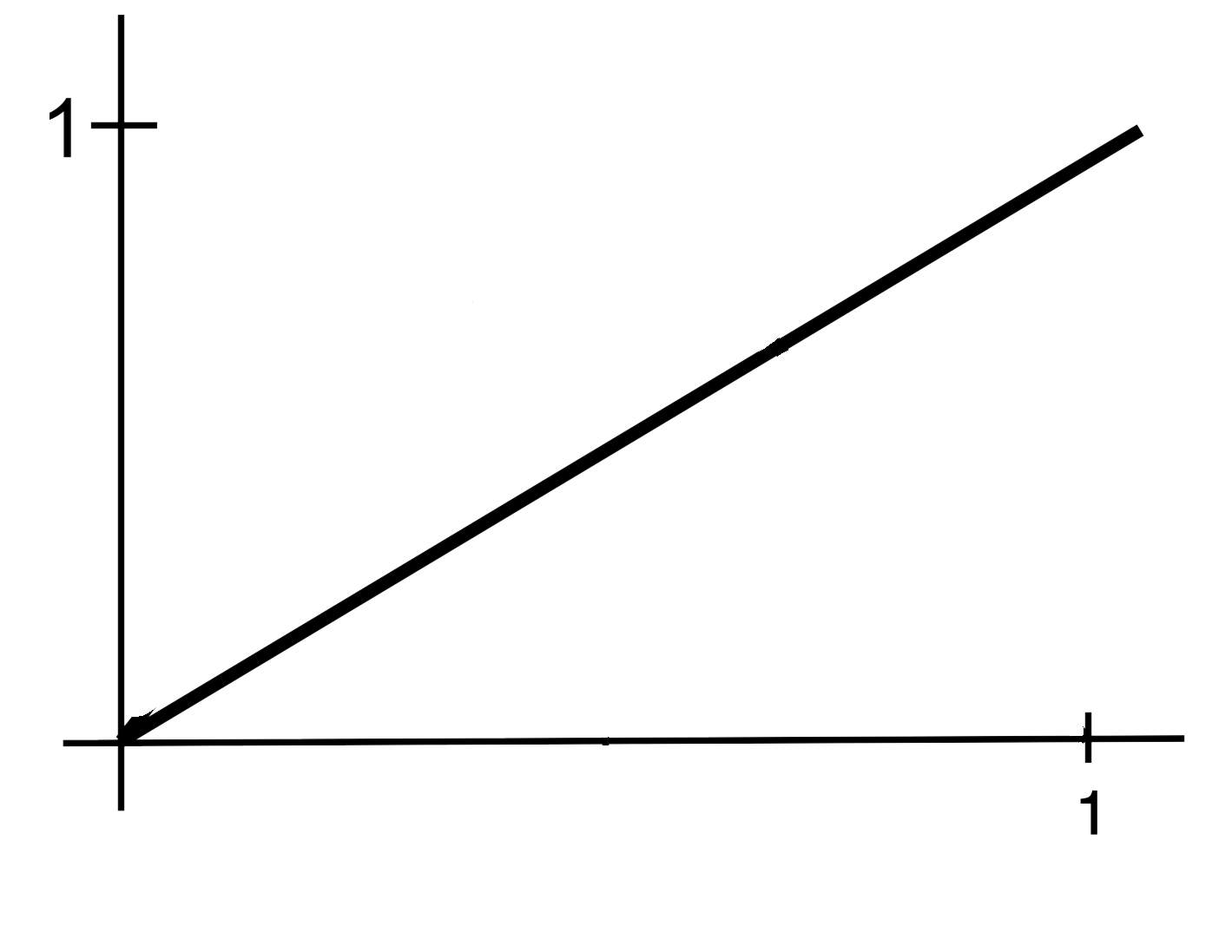}
        Level 0: $\phi_1$
    \end{subfigure}	
   \begin{subfigure}{.22\textwidth}
        \centering
        \includegraphics[width=\linewidth]{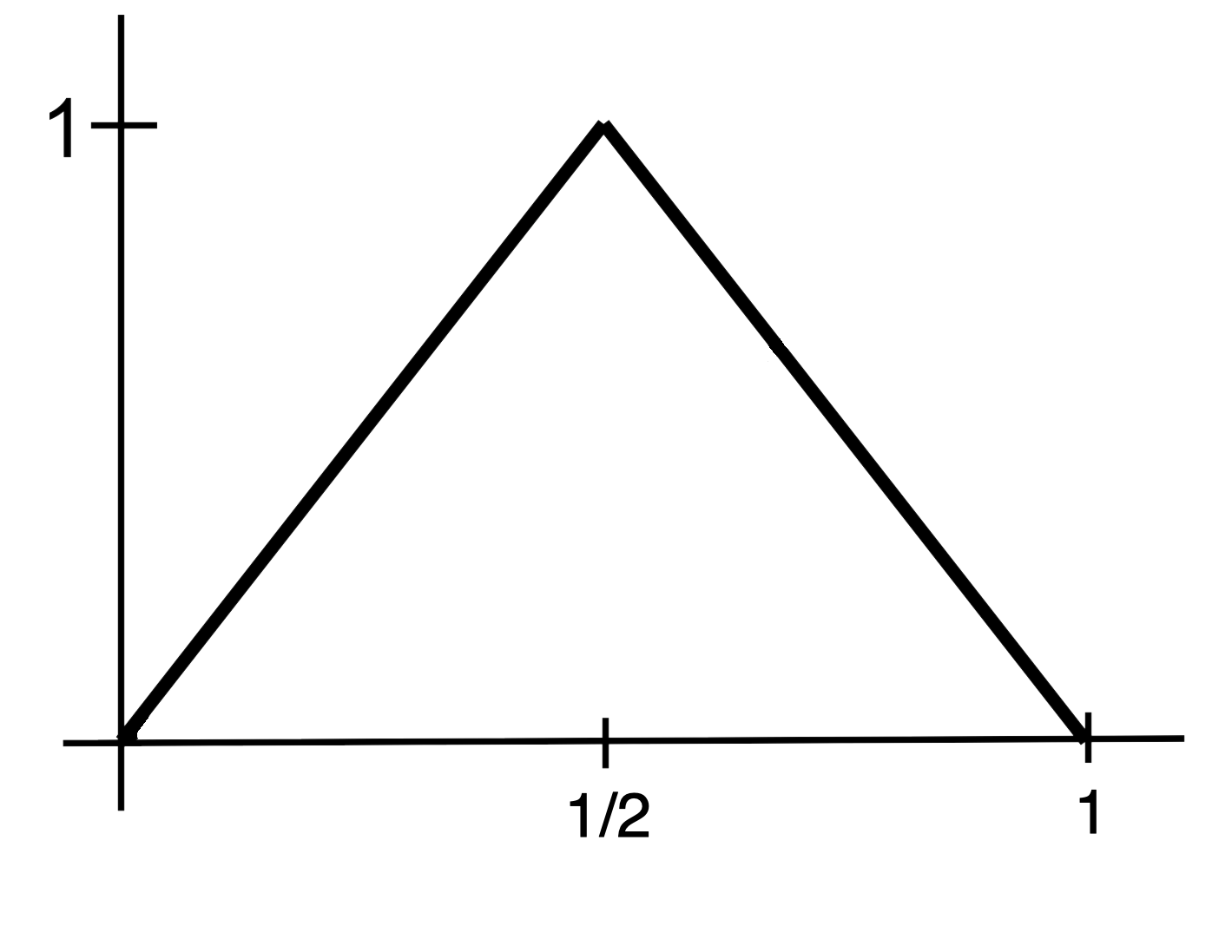}
        Level 1: $\phi_2$
    \end{subfigure}	
    \begin{subfigure}{.22\textwidth}
        \centering
        \includegraphics[width=\linewidth]{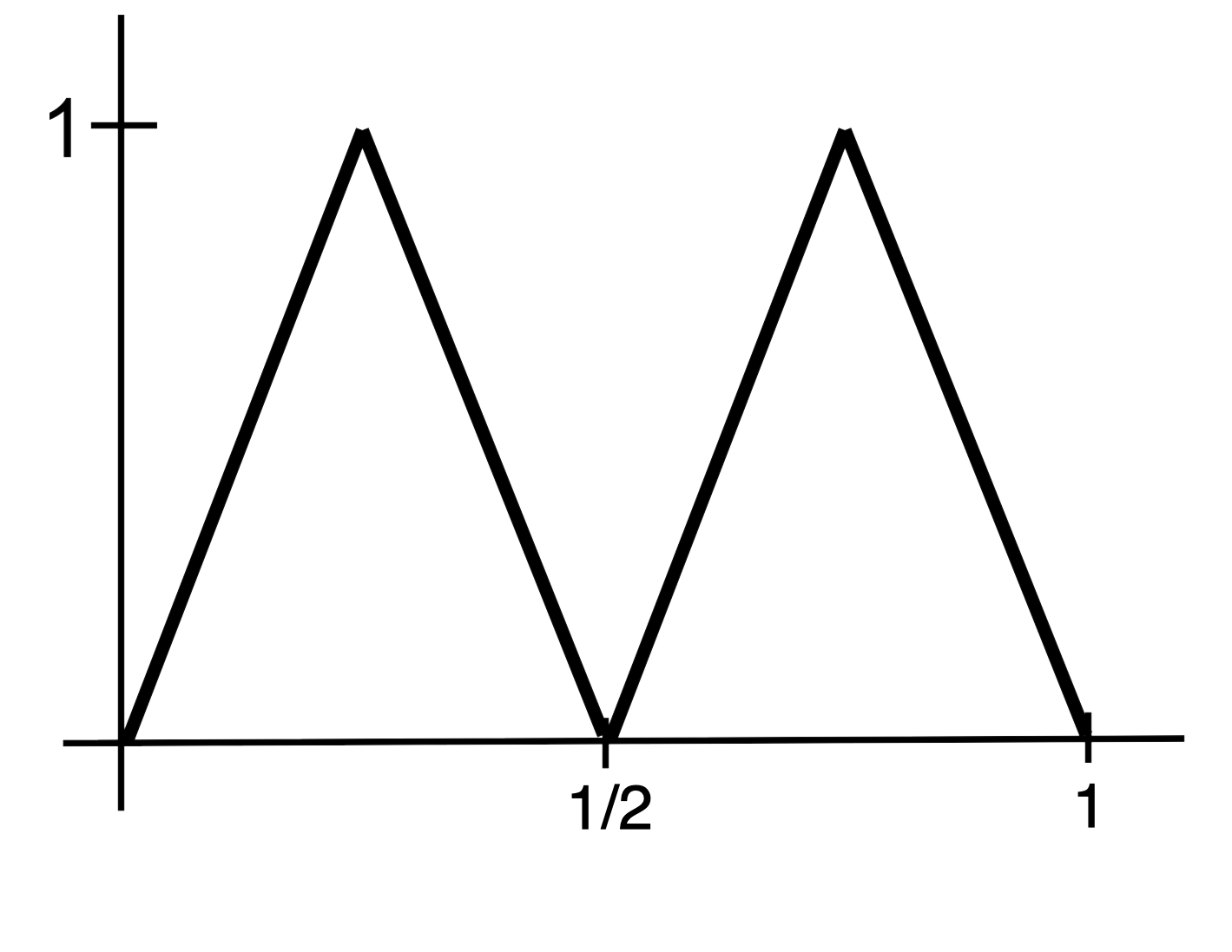}
                Level 2: $\phi_3, \phi_4$
    \end{subfigure}		
      \begin{subfigure}{.22\textwidth}
        \centering
        \includegraphics[width=\linewidth]{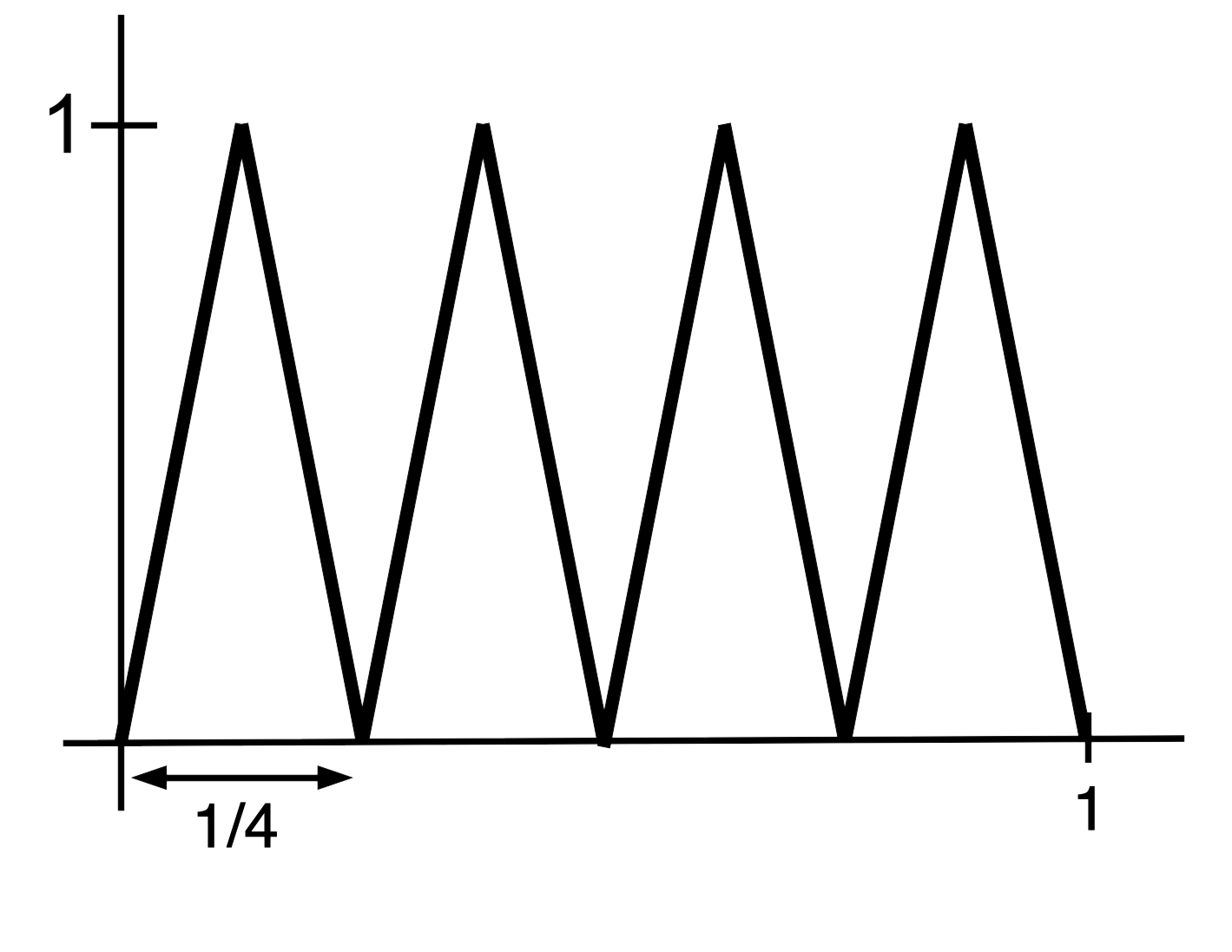}
                Level 3: $\phi_5, \ldots, \phi_8$
    \end{subfigure}		
        \caption{Faber-Schauder system}
  \label{fig: schauder}
\end{figure}

L\'evy's definition of the standard Brownian motion on the interval $[0,1]$ is
\begin{equation}	\label{eq: Levy}
B_n(t) = \sum_{j=1}^\infty G_j \phi_j(t),
\end{equation}
where $G_j$ are independent normal random variables, namely
$$
G_1 \sim N(0,1); \quad G_j \sim N \big( 0, 2^{-\ell(j)-1} \big), \, j=2,3,\ldots.
$$

To construct a superregular random walk, we 
replace the Gaussian weights $G_j$ by Laplacian\footnote{A Laplacian random variable $X \sim \Lap(\l)$ is defined by $\Pr{\abs{X}>t} = \exp(-t/\l)$, $t \ge 0$. 
The Laplacian distribution has density $(1/2\l) \exp(-\abs{x}/\l)$ on $\R$. 
The mean equals zero and the variance equals $2\l^2$.} weights $\Lambda_j \sim \Lap(\log n)$. 
and we truncate the series at $n$. Thus, we set 
\begin{equation}	\label{eq: Levy modified}
W_n(t) = \sum_{j=1}^n \Lambda_j \phi_j(t).
\end{equation}
The superregular random walk could then be defined as 
\begin{equation}	\label{eq: superregular walk definition}
Z_1+\cdots+Z_k = W_n(k/n), 
\quad k=1,\ldots,n.
\end{equation}

\subsection{Formal construction}\label{FCsubsection}

First observe that the regularity property \eqref{eq:regularity} of a probability distribution on $\mathbb{R}^{n}$ passes on to the marginal distributions. 
For example, regularity of a random vector $(X_1,X_2) \in \R^2$ means that 
$$
f_{(X_1,X_2)}(x_1,x_2) \le \exp(-\abs{x_1-y_1}-\abs{x_2-y_2}) \, f_{(X_1,X_2)}(y_1,y_2),
$$
for all $(x_1,y_1),(x_2,y_2)\in\mathbb{R}^{2}$. In particular,
$$
f_{(X_1,X_2)}(x_1,x_2) \le \exp(-\abs{x_2-y_2}) \, f_{(X_1,X_2)}(x_1,y_2).
$$ 
Taking integral with respect to $x_1$ on both sides yields 
$$
f_{X_2}(x_2) \le \exp(-\abs{x_2-y_2}) \, f_{X_2}(y_2),
$$ 
which is equivalent to the regularity of the random vector $X_2 \in \R^1$.
The same argument works in higher dimensions.

Thus, by dropping at most $n/2$ terms if necessary, we can assume without loss of generality 
that 
\begin{equation}	\label{eq: L}
n=2^L \quad \text{for some } L \in \N, 
\end{equation}
Thus, the Faber-Schauder functions $\phi_1,\ldots,\phi_n$ are partitioned in $L+1$ full levels
$0,1,\ldots,L$.
Consider i.i.d.\ random variables 
\begin{equation}	\label{eq: Lambdas}
\Lambda_1,\ldots,\Lambda_n \sim \Lap(2L+1),
\end{equation}
define the random process $W_n(t)$ by equation \eqref{eq: Levy modified}, 
and define the random variables $Z_1,\ldots,Z_k$ 
by equation \eqref{eq: superregular walk definition}.

The construction is complete. 
It remains to check boundedness and regularity.

\subsection{Boundedness}

Fix $k \in [n]$. We have
\begin{align} 
\E(Z_1+\cdots+Z_k)^2
  &= \E \left( \sum_{j=1}^n \Lambda_j \, \phi_j(k/n) \right)^2
  	\quad \text{(by definition)} \nonumber\\
  &= \sum_{j=1}^n \E[\Lambda_j^2] \, \phi_j(k/n)^2	\label{eq: by independence}
  	\quad \text{(by independence and mean zero)} \\
  &= 2(2L+1)^2 \sum_{j=1}^n \phi_j(k/n)^2
  	\quad \text{(by \eqref{eq: Lambdas}).} \nonumber
\end{align}
By construction, the Faber-Schauder functions $\phi_j$ on each given level have disjoint support. Thus, on each level there can be be at most one function that makes the value  
$\phi_j(k/n)^2$ nonzero. By construction, this value is bounded by $1$. 
Adding these values for the $L+1$ levels, we conclude that $\sum_{j=1}^n \phi_j(k/n)^2$ is bounded by $L+1$. Hence
$$
\E(Z_1+\cdots+Z_k)^2 
\le 2(2L+1)^2 (L+1) \lesssim \log^3 n
$$
where we used \eqref{eq: L} in the last step.

\subsection{Regularity}

By definition \eqref{eq: Levy modified} and \eqref{eq: superregular walk definition}, we have 
$$
Z_k = W_n \left( \frac{k}{n} \right) - W_n \left( \frac{k-1}{n} \right)
= \sum_{j=1}^n \Lambda_j \psi_j(k)
$$
where
$$
\psi_j(k) = \phi_j \left( \frac{k}{n} \right) - \phi_j \left( \frac{k-1}{n} \right),
\quad k=1,\ldots,n.
$$
The discrete functions $\psi_j$ can be thought as (discrete) derivatives of the Faber-Schauder functions $\phi_j$, and they are known as (discrete, rescaled) {\em Haar system}, cf.~\cite{semadeni2006schauder,bhattacharya2007basic}. The Haar basis is illustrated in Figure~\ref{fig: haar}.

\begin{figure}[htp]	
     \begin{subfigure}{.22\textwidth}
        \centering
        \includegraphics[width=\linewidth]{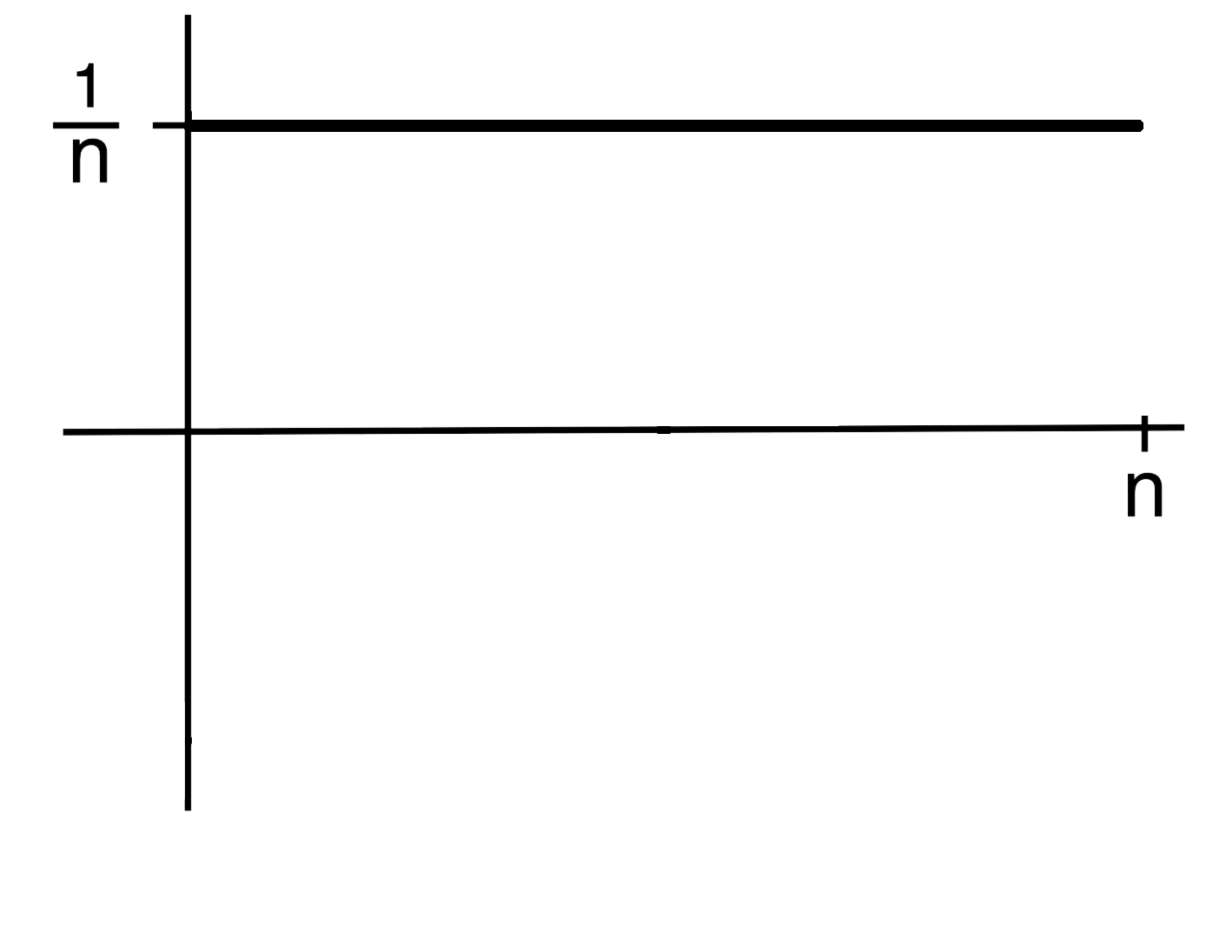}
        Level 0: $\phi_1$
    \end{subfigure}	
   \begin{subfigure}{.22\textwidth}
        \centering
        \includegraphics[width=\linewidth]{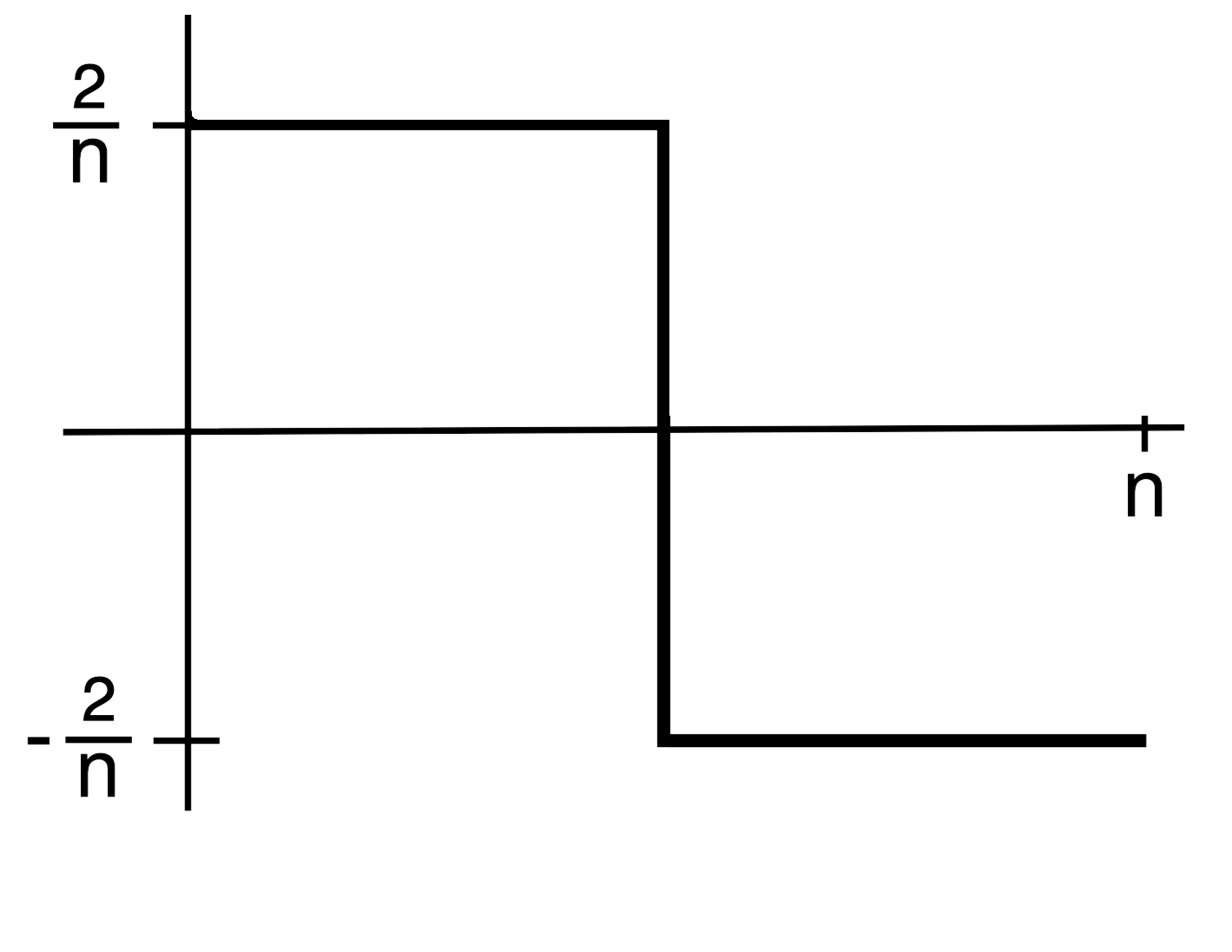}
        Level 1: $\phi_2$
    \end{subfigure}	
    \begin{subfigure}{.22\textwidth}
        \centering
        \includegraphics[width=\linewidth]{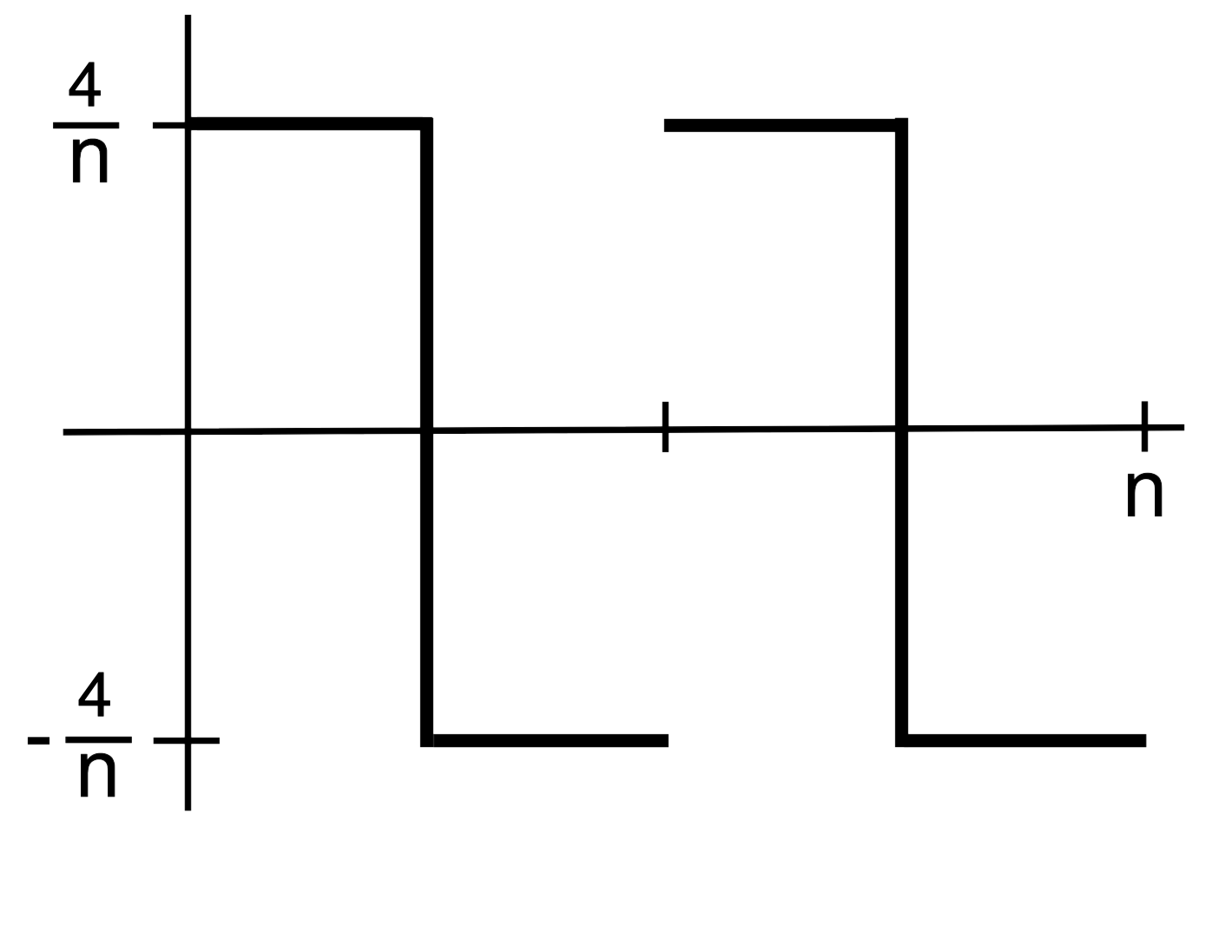}
                Level 2: $\phi_3, \phi_4$
    \end{subfigure}		
      \begin{subfigure}{.22\textwidth}
        \centering
        \includegraphics[width=\linewidth]{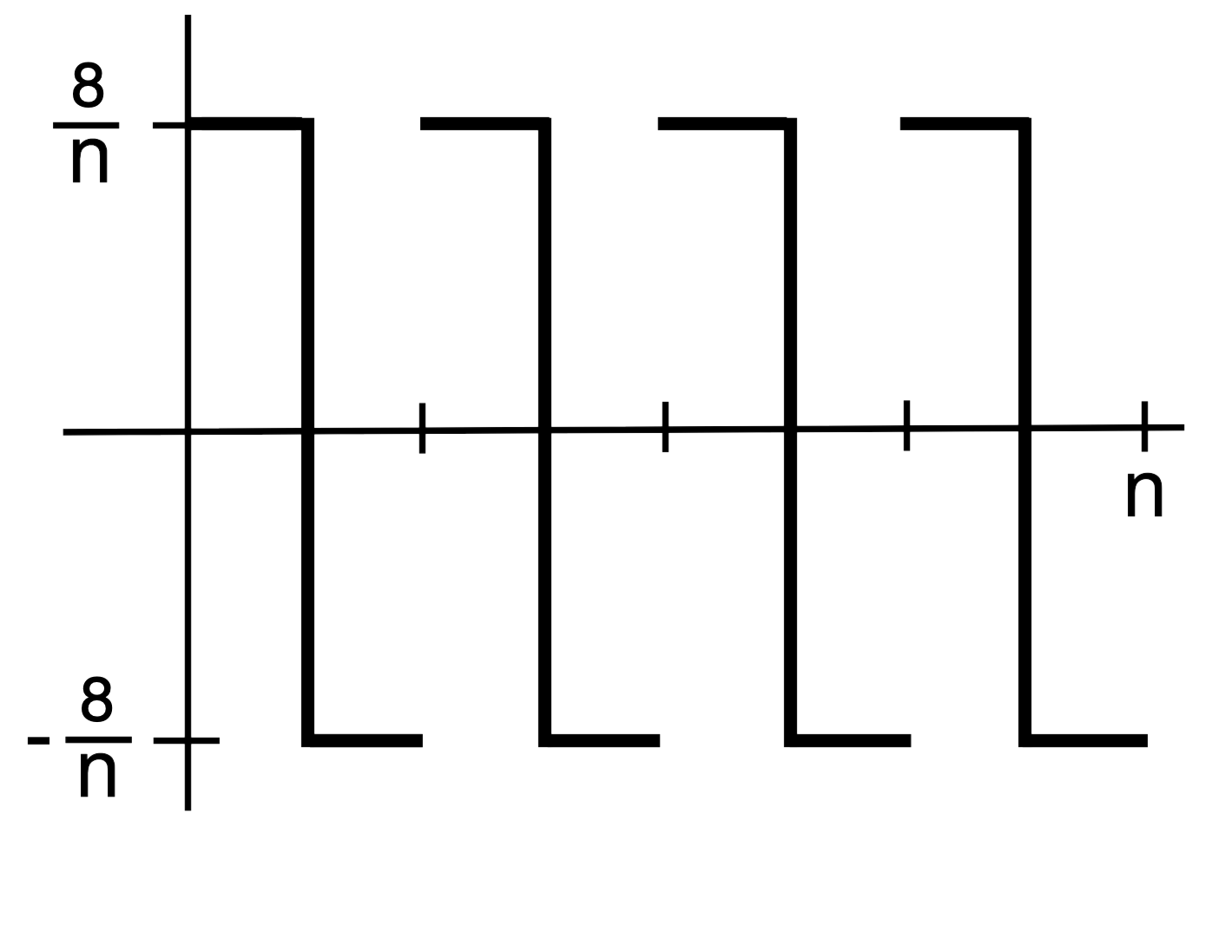}
                Level 3: $\phi_5, \ldots, \phi_8$
    \end{subfigure}		
        \caption{Haar system}
  \label{fig: haar}
\end{figure}

The Haar system $\psi_1,\ldots,\psi_n$ form an orthogonal basis of $\ell^2[n]$, see~\cite{bhattacharya2007basic}. Thus, every function $x \in \ell^2[n]$ admits the orthogonal decomposition
$$
x = \sum_{j=1}^n \l(x)_j \, \psi_j
\quad \text{where } \l(x)_j  = \frac{\ip{\psi_j}{x}}{\norm[0]{\psi_j}_2^2},
$$
where $\langle\,,\rangle$ and $\|\,\|_2$ are the standard inner product and norm on $\ell^{2}[n]$.

The key property of the coefficient vector $\lambda(x)$ is its
approximate sparsity, which we can express via the $\ell^1$ norm. 

\begin{lemma}[Sparsity]				\label{lem: sparsity}
  For any function $x \in \ell^2[n]$, the coefficient vector $\lambda(x)$ satisfies 
  $$
  \norm{\lambda(x)}_1 \le (2L+1) \norm{x}_1.
  $$
\end{lemma}

\begin{proof}
First, let us prove the lemma for the indicator of any single point $k \in [n]$, i.e. for 
$x=\one_{\{k\}}$. Here we have
$$
\l(x)_j = \frac{\psi_j(k)}{\norm[0]{\psi_j}_2^2}.
$$

First, consider $j=1$, the only index on level $\ell=0$. The function $\psi_1(t)=1/n$ 
trivially satisfies $\psi_1(k)=1/n$ and $\norm[0]{\psi_j}_2^2 = 1/n$, 
so we have $\l(x)_1=1$.

Next, consider an index $j$ on some level $\ell \ge 1$. 
By construction, any function $\psi_j$ on that level can takes on three values: 
$0$ and $\pm 2^\ell/n$. Moreover, $\psi_j$ is supported on an interval of length $n/2^{\ell-1}$, see \eqref{eq: Schauder support}. Hence 
$\norm[0]{\psi_j}_2^2 = 2^{\ell+1}/n$, so $\abs{\l(x)_j} \le 2$.

Moreover, the functions $\psi_j$ on any given level have disjoint support. 
So among all such functions on each level, at most one can make the value 
$\psi_j(k)$ and thus $\l(x)_j$ nonzero. As we just showed, for level $0$ this value is $1$, 
and for each subsequent level $\ell \in \{1,\ldots,L\}$, this value is bounded by $2$. 
Summing over all levels, we conclude that 
$\norm{\lambda(x)}_1 \le 2L+1$.

To extend this bound to a general function $x \in \ell^2[n]$, decompose it as 
$x = \sum_{k=1}^n x(k) \one_{\{k\}}$. Then, by linearity, 
$\l(x) = \sum_{k=1}^n x(k) \l(\one_{\{k\}})$, so 
$$
\norm{\l(x)}_1 \le \sum_{k=1}^n \abs{x(k)} \, \norm[1]{\l(\one_{\{k\}})}_1.
$$
The bound $\norm[1]{\l(\one_{\{k\}})}_1 \le 2L+1$ from the first part of 
the argument completes the proof of the lemma. 
\end{proof}

We are ready to prove regularity. Consider the random function  
$Z = \sum_{j=1}^n \Lambda_j \psi_j$ constructed in Subsection \ref{FCsubsection}. 
In our new notation, the coefficient vector of $Z$ is $\l(Z)=(\Lambda_1,\ldots,\Lambda_n)=:\Lambda$.
We have for any $x,y \in \ell^2[n]$:
\begin{equation}	\label{eq: density ratio}
r(x,y) \coloneqq \frac{\dens_X(x)}{\dens_X(y)}
= \frac{\dens_\Lambda(\l(x))}{\dens_\Lambda(\l(y))}.
\end{equation}
To see this, recall that the map $x \mapsto \l(x)$ is a linear bijection on $\ell^2[n]$. 
Hence for any $\e > 0$ and for the unit ball $B$ of $\ell^2[n]$, we have
$$
\frac{\Pr{\textcolor{blue}{Z} \in x+\e B}}{\Pr{\textcolor{blue}{Z} \in y+\e B}}
= \frac{\Pr{\Lambda \in \l(x)+\e \l(B)}}{\Pr{\Lambda \in \l(y)+\e \l(B)}}. 
$$
Taking the limit on both sides as $\e \to 0_+$ and applying the Lebesgue differentiation theorem yield \eqref{eq: density ratio}.

By construction, the coefficients $\Lambda_i$ of the random vector $\Lambda \in \R^n$ are $\Lap(2L+1)$ i.i.d.\ random variables. Hence 
$$
\dens_\Lambda(z) = \frac{1}{\big( 2(2L+1) \big)^n} \exp \Big( -\frac{\norm{z}_1}{2L+1} \Big), 
\quad z \in \R^n.
$$
Thus, 
$$
r(x,y) = \exp \Bigg( \frac{\norm{\l(y)}_1 - \norm{\l(x)}_1}{2L+1} \Bigg).
$$
By the triangle inequality and Lemma~\ref{lem: sparsity}, we have 
$$
\norm{\l(y)}_1 - \norm{\l(x)}_1 
\le \norm{\l(x)-\l(y)}_1
\le (2L+1) \norm{x-y}_1.
$$
Thus
$$
r(x,y) \le \exp(\norm{x-y}_1).
$$

If we express the density in the form $\dens_X(x) = \frac{1}{\beta} e^{-V}$, the bound we proved can be written as 
$$
\exp \left( V(y)-V(x) \right) \le \exp(\norm{x-y}_1), 
$$
or
$V(y)-V(x) \le \norm{x-y}_1$.
Swapping $x$ with $y$ yields $\abs{V(x)-V(y)} \le \norm{x-y}_1$. 
The proof of Theorem~\ref{thm: bounded random walk} is complete.
\qed

\begin{remark}[Boundedness of paths]
  One can easily upgrade the bound \eqref{eq:boundedness2}, which holds in expectation,
  into a concentration bound that holds with high probability. To do so, instead of 
  applying the additivity of variance in \eqref{eq: by independence}, one can use 
  a concentration inequality for sums of independent random variables, 
  e.g. Bernstein's. Moreover, combining the resulting high-probability bound 
  with a union bound, one can obtain boundedness of the entire paths of the 
  random walk, showing that 
  $$
  \E \max_{1 \le k \le n} \big( Z_1+\cdots+Z_k \big)^2 \le C \log^4 n.
  $$
  Since we do not need this result for our application, we leave it to the interested reader.
\end{remark}

\subsection{Beyond the $\ell^1$ norm?}

One may wonder why specifically the $\ell^1$ norm appears in the regularity property of Theorem~\ref{thm: bounded random walk}. As we will see shortly, the regularity with respect to the $\ell^1$ norm is exactly what is needed in our applications to privacy. However, it might be interesting to see if there are natural extensions of Theorem~\ref{thm: bounded random walk} for general $\ell^p$ norms.   The lemma below rules out one such avenue, showing that if a potential $V$ is Lipschitz with respect to the $\ell^p$ norm for some $p>1$, the corresponding random walk deviates at least {\em polynomially} fast (as opposed to logarithmically fast).

\begin{proposition}[No boundedness for $\ell^p$-regular potentials]
\label{prop:no}
  Let $n \in \N$ and consider a probability density of the form 
  $f(z)=\frac{1}{\beta}e^{-V(z)}$ on $\mathbb{R}^{n}$. 
  Assume that the potential $V$ is $1$-Lipschitz in the $\ell^p$-norm. 
  Then a random vector $Z = (Z_1,\ldots,Z_n)$ distributed according to the density $f$
  satisfies
  $$
  \E \abs{Z_1+\cdots+Z_n} 
  \ge \frac{1}{4} n^{1-\frac{1}{p}}.
  $$
\end{proposition}

\begin{proof}
We can write $Z_1+\cdots+Z_n = \ip{Z}{u}$ where $u = (1,\ldots,1)^\tran$. 
Since $\norm[0]{n^{-\frac{1}{p}}u}_p=1$ and $V$ is 1-Lipschitz in the $\ell^{p}$ norm, the densities of the random vectors $Z+n^{-\frac{1}{p}}u$ and $Z$ differ by a multiplicative factor of at most $e$ pointwise. Therefore, 
\begin{align*} 
\E \abs{\ip{Z}{u}}
  &\ge e^{-1} \E \abs[1]{\ip{Z+n^{-\frac{1}{p}}u}{u}} \\
  &\ge e^{-1} \Big( \abs[1]{\ip{n^{-\frac{1}{p}}u}{u}} - \E \abs{\ip{Z}{u}} \Big)
  	\quad \text{(by triangle inequality)} \\
  &= e^{-1} \Big( n^{1-\frac{1}{p}} - \E \abs{\ip{Z}{u}} \Big).
\end{align*}
Rearranging the terms, we deduce that 
$$
\E \abs{\ip{Z}{u}} 
\ge \frac{e^{-1}}{1+e^{-1}} n^{1-\frac{1}{p}}
\ge \frac{1}{4} n^{1-\frac{1}{p}},
$$
which completes the proof.
\end{proof}

In light of Theorem~\ref{thm: bounded random walk} and Proposition~\ref{prop:no}
it might be interesting to see if an obstacle remains for the density $f(z)=\frac{1}{\beta}e^{-V(z)^p}$ for $p>1$.

\section{Metric privacy}			\label{s: MP}

\subsection{Private measures}		\label{s: private measures}

The superregular random walk we just constructed will become the main tool 
in solving the following {\em private measure problem}.
We are looking for a private and accurate algorithm $\AA$ that transforms a probability measure $\mu$ on a metric space $(T,\rho)$ into another finitely-supported probability measure $\AA(\mu)$ on $(T,\rho)$. 

We need to specify what we mean by privacy and accuracy here. 
{\em Metric privacy} offers a natural framework for our problem. 
Namely, we consider Definition~\ref{def: metric privacy} for the space $(\MM(T), \TV)$ of all probability measures on $T$ equipped with the TV metric (recalled in Section~\ref{s: TV}).  Thus, for any pair of input measures $\mu$ and $\mu'$ on $T$ that are close in the TV metric, we would like the distributions of the (random) output measures $\AA(\mu)$ and $\AA(\mu')$ to be close:
\begin{equation}	\label{eq: metric privacy measures}
\frac{\Pr{\AA(\mu) \in S}}{\Pr{\AA(\mu') \in S}} \le \exp \left( \a \norm[0]{\mu-\mu'}_\TV \right).
\end{equation}  
The accuracy will be measured via the {\em Wasserstein distance} (recalled in Section~\ref{s: Wasserstein}). We hope to make $W_1(\AA(\mu),\mu)$ as small as possible. 
The reason for choosing $W_1$ as distance is that it allows us to derive accuracy guarantees for general Lipschitz statistics, as outlined below.

\subsection{Synthetic data}
The private measure problem has an immediate application for {\em differentially private synthetic data}.
Let $(T,\rho)$ be a compact metric space. 
We hope to find an algorithm $\BB$ that transforms the true data 
$X =(X_1,\ldots,X_n) \in T^n$ into synthetic data $Y=(Y_1,\ldots,Y_m) \in T^m$ for some $m$
such that the empirical measures 
$$
\mu_X = \frac{1}{n} \sum_{i=1}^n \d_{X_i}
\quad \text{and} \quad 
\mu_Y = \frac{1}{m} \sum_{i=1}^m \d_{Y_i}
$$
are close in the Wasserstein distance, i.e. we hope to make 
$W_1(\mu_X,\mu_Y)$ small.
This would imply that synthetic data accurately preserves all Lipschitz statistics, i.e. 
$$
\frac{1}{n} \sum_{i=1}^n f(X_i) \approx \frac{1}{m} \sum_{i=1}^m f(Y_i)
$$
for any Lipschitz function $f:T \to \R$.

This goal can be immediately achieved if we solve a version of the private measure
problem, described in Section~\ref{s: private measures}, with the additional requirement that $\AA(\mu)$ be an empirical measure. Indeed, define the algorithm $\BB$ by feeding the empirical measure $\mu_X$ into $\AA$, i.e. set $\BB(X) = \AA(\mu_X)$. The accuracy follows, and the differential privacy of $\BB$ can be seen as follows.

For any pair $X,X'$ of input data that differ in a single element, the corresponding empirical measures differ by at most $1/n$ with respect to the TV distance, i.e. 
$$
\norm{\mu_X-\mu_{X'}}_\TV \le \frac{1}{n}.
$$
Then, for any subset $S$ in the output space, we can use \eqref{eq: metric privacy measures} to get 
$$
\frac{\Pr{\BB(X) \in S}}{\Pr{\BB(X') \in S}} 
= \frac{\Pr{\AA(\mu_X) \in S}}{\Pr{\AA(\mu_{X'}) \in S}} 
\le \exp \left( \a \, \norm[0]{\mu-\mu'}_\TV \right) 
\le \exp(\a/n).
$$
Thus, if $\alpha = \e n$, the algorithm $\BB$ is $\e$-differentially private.
Let us record this observation formally.

\begin{lemma}[Private measure yields private synthetic data] \label{lem: PM to SD}
  Let $(T,\rho)$ be a compact metric space. 
  Let $\AA$ be an algorithm that inputs a probability measure on $T$, and outputs something.
  Define the algorithm $\BB$ that takes data $X =(X_1,\ldots,X_n) \in T^n$ as an input, 
  creates the empirical measure $\mu_X$ and feeds it into the algorithm $\AA$, i.e. set 
  $\BB(X) = \AA(\mu_X)$. 
  If $\AA$ is $\a$-metrically private in the TV metric and $\a = \e n$, then $\BB$ is $\e$-differentially private.
\end{lemma}

Thus, our main focus from now on will be on solving the private measure problem; private synthetic data will follows as a consequence.

\section{A private measure on the line}\label{s:pm_line}

In this section, we construct a private measure on the interval $[0,1]$. 
Later we will extend this construction to general metric spaces. 

\subsection{Discrete input space}

Let us start with a somewhat restricted goal, and then work toward wider generality. 
In this subsection, we will (a) assume that the input measure $\mu$ is always 
supported on some fixed finite subset 
$$
\Omega = \{\omega_1,\ldots,\omega_n\} 
\quad \text{where } 0 \le \omega_1 \le \cdots \le \omega_n \le 1
$$  
and 
(b) allow the output $\AA(\mu)$ to be a signed measure.
We will measure accuracy with the Wasserstein distance.

\subsubsection{Perturbing a measure by a superregular random walk}
Apply the Superregular Random Walk Theorem~\ref{thm: bounded random walk} and rescale the random variables $Z_i$ by setting $U_i = (2/\a) Z_i$. The regularity property of the random vector 
$U = (U_1,\ldots,U_n)$ takes the form 
\begin{equation}	\label{eq: U regularity}
\frac{\dens_U(x)}{\dens_U(y)} \le \exp \left( \frac{\a}{2} \norm{x-y}_1 \right)
\quad \text{for all } x,y \in \R^n,
\end{equation}
and the boundedness property implies that
\begin{equation}        \label{eq: Y boundedness}
\max_{1 \le k \le n}\E  \abs{U_1+\cdots+U_k} \le \frac{C \log^{\frac{3}{2}} n}{\a}.
\end{equation}


Let us make the algorithm $\AA$ perturb the measure $\mu$ on $\Omega$ by the weights $U_i$, 
i.e. we set 
\begin{equation}	\label{eq: perturbation}
\AA(\mu)(\omega_i) = \mu(\{\omega_i\})+U_i, 
\quad i=1,\ldots,n.
\end{equation}

\subsubsection{Privacy}
Any measure $\nu$ on $\Omega$ can be identified with the vector $\bar{\nu} \in \R^n$
by setting $\bar{\nu}_i = \nu(\{\omega_i\})$. 
Then, for any measure $\eta$ on $\Omega$, we have
\begin{equation}	\label{eq: measure to vector}
\dens_{\AA(\mu)}(\eta) 
= \dens_{\bar{\mu}+U}(\bar{\eta})
= \dens_{U}(\bar{\eta}-\bar{\mu}).
\end{equation}
Fix two measures $\mu$ and $\mu'$ on $\Omega$. By above, we have 
\begin{align*} 
\frac{\dens_{\AA(\mu)}(\eta)}{\dens_{\AA(\mu')}(\eta)}
  &= \frac{\dens_{U}(\bar{\eta}-\bar{\mu})}{\dens_{U}(\bar{\eta}-\bar{\mu'})}
	\quad \text{(by \eqref{eq: measure to vector})}\\
  &\le \exp \left(\frac{\a}{2} \norm{\bar{\mu}-\bar{\mu'}}_1 \right)
	\quad \text{(by \eqref{eq: U regularity})}\\
  &= \exp \left( \a \norm{\mu-\mu'}_\TV \right)
	\quad \text{(by \eqref{eq: TV as L1})}.
\end{align*}
This  shows that the algorithm $\AA$ is $\a$-metrically private 
in the TV metric. 

\subsubsection{Accuracy}

By definition definition \eqref{eq: Vallender} of Wasserstein distance for signed measures, we have 
\begin{align*} 
\E W_{1}(\AA(\mu),\mu) 
  &= \int_{0}^{1} \E \abs{\big( \AA(\mu)-\mu \big) \big( [0,x] \big)} \,dx 
  	\quad \text{(using linearity of expectation)} \\
  &= \int_{0}^{1} \E \abs[3]{\sum_{j: \, \omega_j \le x} \big( \AA(\mu)-\mu \big) (\omega_j)} \,dx 
  	 \quad \text{(measures are supported on points $\omega_j$)} \\
  &= \int_{0}^{1} \E \abs[3]{\sum_{j=1}^{k(x)} U_j} \,dx
    	 \quad \text{(by \eqref{eq: perturbation}, where we set $k(x) = \abs{\{j:\, \omega_j \le x\}}$)}\\
  &\le \max_{1 \le k \le n} \E \abs[3]{\sum_{j=1}^k U_j} 
  \le \frac{C \log^{\frac{3}{2}} n}{\a}
  	\quad \text{(by \eqref{eq: Y boundedness})}.
\end{align*}

The following result summarizes what we have proved.

\begin{proposition}[Input in discrete space, output signed measure]     \label{prop: signed}
  Let $\Omega$ be finite subset of $[0,1]$ and let $n = \abs{\Omega}$.
  Let $\a>0$.
  There exists a randomized algorithm $\AA$ that takes a
  probability measure $\mu$ on $\Omega$ as an input
  and returns a signed measure $\nu$ on $\Omega$ as an output,
  and with the following two properties.
  \begin{enumerate}[(i)]
    \item (Privacy): the algorithm $\AA$ is $\a$-metrically private
      in the TV metric.
    \item (Accuracy): for any input measure $\mu$,
    the expected accuracy of the output signed measure $\nu$
    in the Wasserstein distance is
        $$
        \E W_{1}(\nu,\mu) \le \frac{C \log^{\frac{3}{2}} n}{\a}.
        $$
  \end{enumerate}
\end{proposition}

Let $\nu$ be the signed measure obtained in Proposition \ref{prop: signed}. Let $\widehat{\nu}$ be a probability measure on $\Omega$ that minimizes $W_{1}(\widehat{\nu},\nu)$. In view of \eqref{eq: Vallender}, finding $\widehat{\nu}$ can be cast as convex problem, although the minimizer may not be unique. 
By minimality, $W_{1}(\widehat{\nu},\nu)\leq W_{1}(\mu,\nu)$. So $W_{1}(\widehat{\nu},\mu)\leq W_{1}(\widehat{\nu},\nu)+W_{1}(\nu,\mu)\leq 2W_{1}(\mu,\nu)$.
Thus, we can upgrade the previous result, making the output a measure (as opposed to signed measure):

\begin{proposition}[Private measure on a finite subset of the interval]                 \label{prop: private discrete interval}
  Let $\Omega$ be finite subset of $[0,1]$ and let $n = \abs{\Omega}$.
  Let $\a>0$.
  There exists a randomized algorithm $\BB$ that takes a
  probability measure $\mu$ on $\Omega$ as an input
  and returns a probability measure $\nu$ on $\Omega$ as an output,
  and with the following two properties.
  \begin{enumerate}[(i)]
    \item (Privacy): the algorithm $\BB$ is $\a$-metrically private
      in the TV metric.
    \item (Accuracy): for any input measure $\mu$,
    the expected accuracy of the output measure $\nu$
    in the Wasserstein distance is
        $$
        \E W_{1}(\nu,\mu) \le \frac{C \log^{\frac{3}{2}} n}{\a}.
        $$
  \end{enumerate}
\end{proposition}

\subsection{Extending the input space to the interval} \label{s: discrete to general}

Next, we would like to extend our framework to a continuous setting, 
and allow measures to be supported by the entire interval $[0,1]$. 
We can do this by quantization. 

\subsubsection{Quantization}
Fix $n \in \NN$ and let $\NN = \{\omega_1,\ldots,\omega_n\}$ be a $(1/n)$-net of $[0,1]$.
Consider the proximity partition 
$$
[0,1] = I_1 \cup \cdots \cup I_n
$$
where we put a point $x \in [0,1]$ into $I_i$ if $x$ is closer to $\omega_i$ that to any other points in $\NN$. (We break any ties arbitrarily.) 

We can quantize any signed measure $\nu$ on $[0,1]$ by 
defining
\begin{equation}	\label{eq: subintervals}
\nu_\NN \left( \{\omega_i\} \right) = \nu(I_i), \quad i=1,\ldots,n.
\end{equation}
Obviously, $\nu_\NN$ is a signed measure on $\NN$. Moreover, if $\nu$ is a measure, then so is $\nu_\NN$. And if $\nu$ is a probability measure, then so is $\nu_\NN$. In the latter case, it follows from the construction that 
\begin{equation}	\label{eq: W quantization}
W_1(\nu,\nu_\NN) \le 1/n. 
\end{equation}
(By definition of the net, transporting any point $x$ to the closest point $\omega_i$ covers distance at most $1/n$.)

\begin{lemma}[Quantization is a contraction in TV metric] 	\label{lem: quantization contraction}
  Any signed measure $\nu$ on $[0,1]$ satisfies 
  $$
  \norm{\nu_\NN}_\TV \le \norm{\nu}_\TV.
  $$
\end{lemma}

\begin{proof}
Using \eqref{eq: TV as L1}, \eqref{eq: subintervals}, and \eqref{eq: TV}, we obtain
$$
\norm{\nu_\NN}_\TV
=\frac{1}{2} \sum_{i=1}^n \abs{ \nu_\NN(\{\omega_i\})}
= \frac{1}{2} \sum_{i=1}^n \abs{\nu(I_i)} 
\le \norm{\nu}_\TV.
$$
The lemma is proved.
\end{proof}

\subsubsection{A private measure on the interval}

\begin{theorem}[Private measure on the interval]		\label{thm: pm on interval}
  Let $\a \ge 2$.
  There exists a randomized algorithm $\AA$ that takes a
  probability measure $\mu$ on $[0,1]$ as an input 
  and returns a finitely-supported probability measure $\nu$ on $[0,1]$
  as an output, and with the following two properties. 
  \begin{enumerate}[(i)]
    \item (Privacy): the algorithm $\AA$ is $\a$-metrically private in the TV metric. 
    \item (Accuracy): for any input measure $\mu$, 
      the expected accuracy of the output measure $\nu$ 
      in the Wasserstein distance is
      $$
      \E W_1 \left( \nu,\mu \right) \le \frac{C \log^{\frac{3}{2}} \a}{\a}.
      $$
  \end{enumerate}
\end{theorem}

\begin{proof}
Take a measure $\mu$ on $[0,1]$, preprocess it by quantizing as in the previous subsection, 
and feed the quantized measure $\mu_\NN$ into the algorithm $\BB$ of Proposition~\ref{prop: private discrete interval} for $\Omega=\NN$.

The contraction property (Lemma~\ref{lem: quantization contraction}) ensures that
$$
\norm{\mu_\NN - \mu'_\NN}_\TV \le \norm{\mu - \mu'}_\TV.
$$
This and the privacy property of Proposition~\ref{prop: private discrete interval} 
for measures on $\NN$ guarantee
that quantization does not destroy privacy, i.e. the algorithm $\mu \mapsto \BB(\mu_\NN)$ is
still $\a$-metrically private as claimed.  

As for the accuracy,  Proposition~\ref{prop: private discrete interval} for the measure $\mu_\NN$ gives 
$$
\E W_1 \left( \BB(\mu_\NN), \mu_\NN \right)
\le \frac{C \log^{\frac{3}{2}} n}{\a}.
$$
Moreover, the accuracy of quantization \eqref{eq: W quantization} states that 
$W_1(\mu,\mu_\NN) \le 1/n$. By triangle inequality, we conclude that 
$$
\E W_1 \left( \BB(\mu_\NN), \mu \right)
\le \frac{1}{n} + \frac{C \log^{\frac{3}{2}} n}{\a}.
$$
Taking $n$ to be the largest integer less than or equal to $\alpha$ yields the conclusion of the theorem.
\end{proof}

\section{The Traveling Salesman Problem}		\label{s: width}

In order to extend the construction of the private measure on the interval $[0,1]$ to a general metric space $(T,\rho)$, a natural approach would be to map the interval $[0,1]$ onto some {\em space-filling curve} of $T$. Since a space filling curves usually are infinitely long, we should do this on the discrete level, for some $\d$-net of $T$ rather than $T$ itself. In this section, we will bound length of such discrete space-filling curve in terms of the metric geometry of $T$. In the next section, we will see how this bound determines the accuracy of a private measure in $T$. 

A natural framework for this step is related to Traveling Salesman Problem (TSP), which is a central problem in optimization and computer science, and whose history goes back to at least 1832 \cite{ABCC}. 

Let $G=(V,E)$ be an undirected weighted connected graph. We occasionally refer 
to the weights of the edges as lengths. 
A {\em tour} of $G$ is a connected walk on the edges that visits every vertex at least once, 
and returns to the starting vertex. The TSP is the problem of finding a tour of $G$ with the shortest length. Let us denote this length by $\TSP(G)$.

Although it is NP-hard to compute TSP(G), or even to approximate it
within a factor of $123/122$ \cite{KLS}, an algorithm of Christofides and Serdyukov \cite{Christofides, Serdyukov} from 1976 gives a $3/2$-approximation for TSP, 
and it was shown recently that the factor $3/2$ can be further improved \cite{KKG}.

\subsection{TSP in terms of the minimum spanning tree}
Within a factor of $2$, the traveling salesman problem is equivalent to another 
key problem, namely the problem of finding the {\em minimum spanning tree} (MST) of $G$. 
A spanning tree of $G$ is a subgraph that is a tree and which includes all vertices of $G$. It always exists and can be found in polynomial time~\cite{kruskal1956shortest,prim1957shortest}.  A spanning tree of $G$ with the smallest length is called the minimum spanning tree of $G$; we denote its length by $\MST(G)$. The following equivalence is a folklore.

\begin{lemma}	\label{lem: TSP MST}
  Any undirected weighted connected graph $G$ satisfies 
  $$
  \MST(G) \le \TSP(G) \le 2 \MST(G).
  $$
\end{lemma}

\begin{proof}
For the lower bound, it is enough to find a spanning tree of $G$ of length bounded by $\TSP(G)$. Consider the minimal tour of $G$ of length $\TSP(G)$ as a subgraph of $G$. Let $T$ be a spanning tree of the tour. Since the tour contains all vertices of $G$, so does $T$, and thus $T$ is a spanning tree of $G$. Since $T$ is obtained by removing some edges of the tour, the length of $T$ is bounded by of the tour, which is $\TSP(G)$. The lower bound is proved.

For the upper bound, note that dropping any edges of $G$ can only increase the value of TSP.
Thus TSP of $G$ is bounded by the TSP of its spanning tree $T$.
Moreover, TSP of any tree $T$ equals twice the sum of lengths of the edges of $T$.
This can be seen by considering the {\em depth-first search} tour of $T$,
which starts at the root and explores as deep as possible along each branch before backtracking, see Figure~\ref{fig: tree-width}. \end{proof}

\begin{figure}[htp]			
  \centering 
  \includegraphics[width=0.5\textwidth]{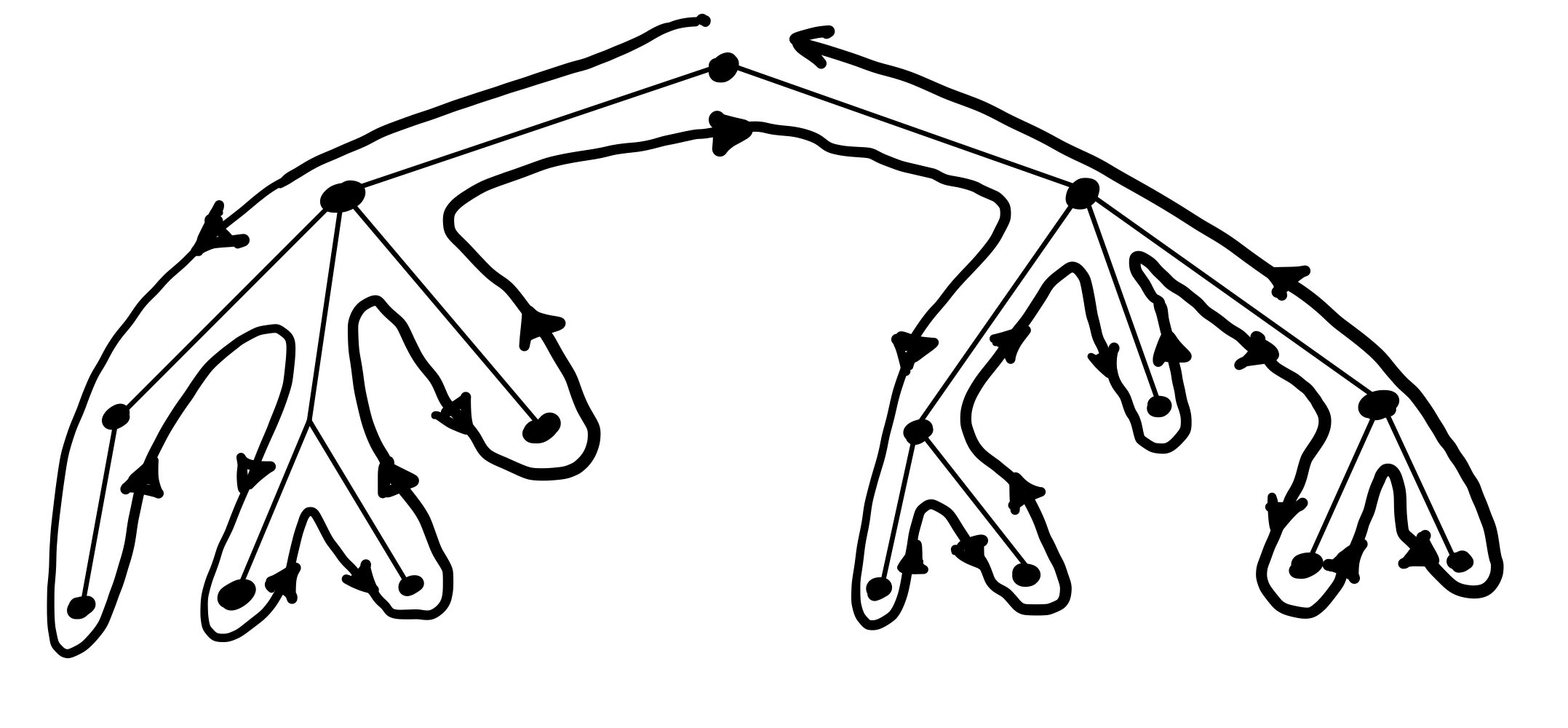} 
  \caption{The depth-first search tour demonstrates that the 
  TSP of a tree equals twice the sum of lengths of its edges.}
  \label{fig: tree-width}
\end{figure}

\subsection{Metric TSP}

Let $(T,\rho)$ be a finite metric space. 
We can consider $T$ as a complete weighted graph, whose weights of edges 
are defined as the distances between the points. The TSP for $(T,\rho)$ is known as {\em metric TSP}. 

Although a tour can visit the same vertex of $T$ multiple times, 
this can be prevented by skipping the vertices previously visited. 
The triangle inequality shows that skipping can only decrease the length of the tour.
Therefore, the shortest tour in a complete graph is always a {\em Hamiltonian cycle}, 
a walk that visits all vertices of $T$ exactly once before returning to the starting vertex. 
Let us record this observation:

\begin{lemma}		\label{lem: Hamiltonian cycles}
  The TSP of a finite metric space $(T,\rho)$ equals the smallest length of a 
  Hamiltonian cycle of $T$.
\end{lemma}

\subsection{A geometric bound on TSP}

We would like to compute $\TSP(T)$ in terms of the geometry of the metric space $(T,\rho)$.
Here we will prove an upper bound on $\TSP(T)$ in terms of the covering numbers. 
Recall that the {\em covering number} $N(T,\rho,\e)$ is defined as the smallest cardinality of an $\e$-net of $T$, or equivalently the smallest number of closed balls with centers in $T$ and radii $\e$ whose union covers $T$, see \cite[Section~4.2]{vershyninbook}.

\begin{theorem}[TSP via covering numbers]		\label{thm: width integral}
  For any finite metric space $(T,\rho)$, we have
  $$
  \TSP(T) \le 16 \int_0^\infty \left( N(T,\rho,x)-1 \right) \, dx.
  $$
\end{theorem}

\begin{proof}
{\em Step 1: constructing a spanning tree.}
Let us construct a small spanning tree $T_0$ of $T$ and use Lemma~\ref{lem: TSP MST}. 
Let $\e_j = 2^{-j}$, $j \in \Z$, and let $\NN_j$ be $\e_j$-nets of $T$ with cardinalities
$\abs{\NN_j} = N(T,\rho,\e_j)$.
Since $T$ is finite, we must have $\abs{\NN_j} = 1$ for all sufficiently small $j$. Let $j_0$ be the largest integer for which $\abs{\NN_{j_0}} = 1$.

At the root of $T_0$, let us put a single point that forms the net $\NN_{j_0}$.
At the next level, put all the points of the net $\NN_{j_0+1}$, and connect them to the root by edges. 
The weights of these edges, which are defined as the distances of the points to the root, 
are all bounded by $\e_{j_0}$. 
At the next level, put all points of the net $\NN_{j_0+2}$, and connect each such point to the closest point in the previous level $\NN_{j_0+1}$. (Break any ties arbitrarily.) Since the latter set is a $\e_{j_0+1}$-net, the weights of all these edges are bounded by $\e_{j_0+1}$. Repeat these steps until the levels do not grow anymore, i.e. until the level contains all the points in $T$; 
see Figure~\ref{fig: chaining} for illustration.

\begin{figure}[htp]			
  \centering 
  \includegraphics[width=0.4\textwidth]{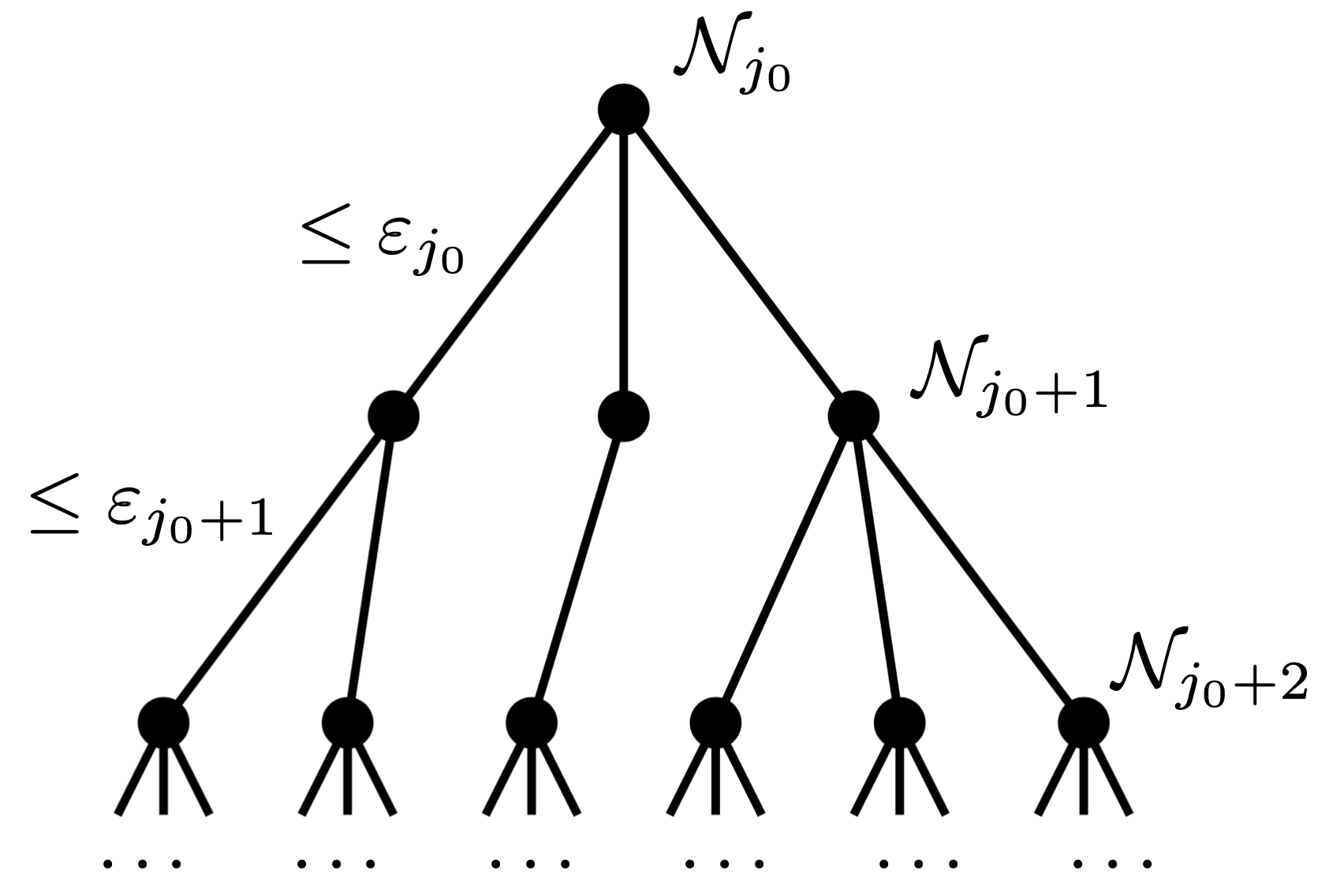} 
  \caption{Chaining: construction of a spanning tree of a metric space.}
  \label{fig: chaining}
\end{figure}

If all the nets $\NN_j$ that make up the levels of the tree $T_0$ are disjoint, then $T_0$ is a spanning tree of $T$. Assume that this is the case for time being.

\smallskip

{\em Step 2: bounding the length of the tree.} 
For each of the levels $j = j_0+1, j_0+2,\ldots$, the tree $T_0$ has $\abs{\NN_j}$ edges connecting the points of level $j$ to the level $j-1$, and each such edge has length (weight) bounded by $\e_{j-1}$. 
So $\MST(T)$ is bounded by the sum of the lengths of the edges of $T_0$, i.e. 
$$
\MST(T) \le \sum_{j=j_0+1}^\infty \e_{j-1} \abs{\NN_j}.
$$

{\em Step 3: bounding the sum by the integral.} 
Our choice $\e_j = 2^{-j}$ yields $\e_{j-1} = 4(\e_j-\e_{j+1})$.
Moreover, our choice of $j_0$ yields $\abs{\NN_j} \ge 2$ for all $j \ge j_0+1$, 
which implies $\abs{\NN_j} \le 2 \left( \abs{\NN_j}-1 \right)$ for such $j$. Therefore
\begin{align}\label{MSTbound}
\MST(T) 
  &\le 8 \sum_{j=j_0+1}^\infty \left( \e_j-\e_{j+1} \right) \left( \abs{\NN_j}-1 \right) \\
  &= 8 \sum_{j=j_0+1}^\infty \int_{\e_{j+1}}^{\e_j} \left( N(T,\rho,\e_j)-1 \right) dx 
  	\quad \text{(since $\abs{\NN_j} = N(T,\rho,\e_j)$)} \nonumber\\
  &\le 8 \int_0^\infty \left( N(T,\rho,x)-1 \right) \, dx.\nonumber
\end{align}
An application of Lemma~\ref{lem: TSP MST} completes the proof.

\smallskip

{\em Step 4: splitting.} 
The argument above assumes that all levels $\NN_j$ of the tree $T_0$ are disjoint.
This assumption can be enforced by splitting the points of $T$. 
If, for example, a point $\omega \in \NN_j$ is also used in $\NN_k$ for some $k < j$, 
add to $T$ another a replica of $\omega$ -- a point $\omega'$ that has zero distance to $\omega$ and the same distances to all other points as $\omega$. Use $\omega$ in $\NN_j$ and $\omega'$ in $\NN_k$.
Preprocessing the metric space $(T,\rho)$ by such splitting yields 
a pseudometric space $(T',\rho)$ in which all levels $\NN_j$ are disjoint,
and whose TSP is the same. 
\end{proof}

\begin{remark}[Integrating up to the diameter]
  Note that $N(T,\rho,x) = 1$ for any $x>\diam(T)$, since any single point makes an $x$-net of $T$ for such $x$. Therefore, the integrand in Theorem~\ref{thm: width integral} vanishes for such $x$, and we have
  \begin{equation}	\label{eq: up to diam}
  \TSP(T) \le 16 \int_0^{\diam(T)} N(T,\rho,x) \, dx.
  \end{equation}
\end{remark}

\subsection{Folding}		\label{s: folding}

It is a simple observation that an interval of length $\TSP(T)$ can be embedded, or ``folded'', into $T$:

\begin{proposition}[Folding]			\label{prop: folding}
  For any finite metric space $(T,\rho)$
  there exists a finite subset $\Omega$ of the interval $[0,\TSP(T)]$
  and a $1$-Lipschitz bijection $F: \Omega \to T$. 
\end{proposition}

Heuristically, the map $F$ ``folds'' the interval $[0,\TSP(T)]$ into
the shortest Hamiltonian path of the metric space $T$, 
see Figure~\ref{fig: folding}. 
We can think of this as a {\em space-filling curve} of $T$.

\begin{figure}[htp]			
  \centering 
  \includegraphics[width=\textwidth]{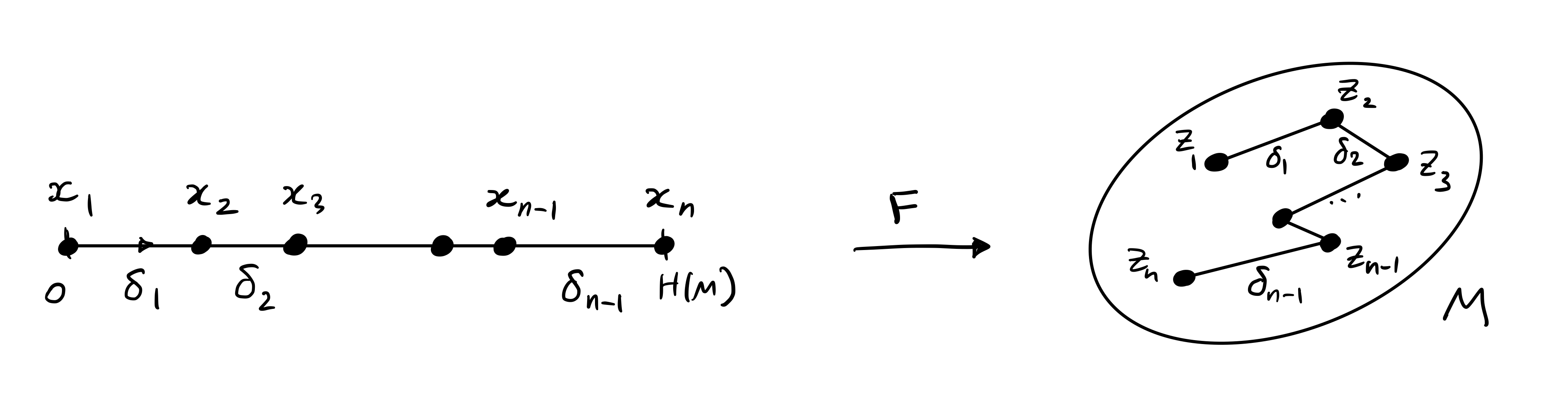} 
  \caption{The map $F$ folds an interval $[0,\TSP(M)]$ into a Hamiltonian path (a ``space-filling curve'') of the metric space $T$.}
  \label{fig: folding}
\end{figure}

\begin{proof}
Let us exploit the heuristic idea of folding.
Fix a Hamiltonian cycle in $T$ of length $\TSP(T)$, whose existence is given by Lemma~\ref{lem: Hamiltonian cycles}. Formally, this means that we can label the elements of the space as $T = \{z_1,\ldots,z_n\}$ in such a way that the lengths 
$$
\d_i = \rho \left( z_{i+1}, z_i \right), \quad i=1,\ldots,n-1,
$$
satisfy 
$\sum_{i=1}^{n-1} \d_i \le \TSP(T)$.
Define $\Omega = \{x_1,\ldots,x_n\}$ by 
$$
x_1=0; \quad x_k = \sum_{i=1}^{k-1} \d_i, \; k=2,\ldots,n.
$$
Then all $x_k \le \TSP(T)$, so $\Omega \subset [0,\TSP(T)]$ as claimed. 

Note that for every $k=1,\ldots,n-1$ we have
$$
\rho \left( z_{k+1}, z_k \right) = \d_k = x_{k+1}-x_k.
$$
Then, for any integers $1 \le k \le k+j \le n$, triangle inequality and telescoping give
\begin{align*} 
\rho \left( z_{k+j}, z_k \right)
  &\le \rho \left( z_{k+j}, z_{k+j-1} \right) + \rho \left( z_{k+j-1}, z_{k+j-2} \right) + \cdots + \rho \left( z_{k+1}, z_k \right) \\
  &= \left( x_{k+j} - x_{k+j-1} \right) + \left( x_{k+j-1} - x_{k+j-2} \right) + \cdots + \left( x_{k+1}- x_k \right)\\
  &=x_{k+j} - x_k.
\end{align*}
This shows that the folding map $F: x_k \mapsto z_k$ is a bijection that satisfies 
$$
\rho \left( F(x), F(y) \right) \le \abs{x-y}
\quad \text{for all } x,y \in \Omega.
$$
In other words, $F$ is $1$-Lipschitz. The proof is complete.
\end{proof}

\section{A private measure on a metric space}\label{s:pm_mp}

We are ready to construct a private measure on an arbitrary compact metric space $(T,\rho)$. We do this as follows: (a) discretize $T$ replacing it with a finite $\d$-net; (b) fold an interval of length $\TSP(T)$ onto $T$ using Proposition~\ref{prop: folding}; and (c) using this folding, pushforward onto $T$ the private measure on the interval constructed in Section~\ref{s:pm_line}. The accuracy of the resulting private measure on $T$ is determined by the length of the interval $\TSP(T)$, which in turn can be expressed using the covering numbers of $T$ (Theorem~\ref{thm: width integral}).

\subsection{Finite metric spaces}

Let us start by extending Proposition~\ref{prop: private discrete interval} from a finite subset on $[0,1]$ to a finite subset of $(T,\rho)$. 

\begin{proposition}[Private measure on a finite metric space]			\label{prop: private finite metric}
  Let $(T,\rho)$ be a finite metric space and let $n = \abs T$. 
  Let $\a>0$.
  There exists a randomized algorithm $\BB$ that takes a probability measure $\mu$ on $T$ as an input and returns a probability measure $\nu$ on $T$ as an output,
  and with the following two properties. 
  \begin{enumerate}[(i)]
    \item (Privacy): the algorithm $\BB$ is $\a$-metrically private 
      in the TV metric.
    \item (Accuracy): for any input measure $\mu$, 
    the expected accuracy of the output measure $\nu$ 
    in the Wasserstein distance is
	$$
	\E W_1(\nu, \mu) \le \frac{C \log^{\frac{3}{2}} n}{\a} \, \TSP(T).
	$$
  \end{enumerate}
\end{proposition}

\begin{proof}
Applying Folding Proposition~\ref{prop: folding}, we obtain an $n$-element subset 
$\Omega \subset [0,\TSP(T)]$ and a $1$-Lipschitz bijection $F: \Omega \to T$. 
Applying Proposition~\ref{prop: private discrete interval} and rescaling by the factor $\TSP(T)$,
we obtain an $\alpha$-metrically private algorithm $\BB$ 
that transforms a probability measure $\mu$ on $\Omega$ 
into a probability measure $\nu$ on $\Omega$, and whose accuracy is
\begin{equation}	\label{eq: W1 prelim}
\E W_1(\nu, \mu)
\le \frac{C \log^{\frac{3}{2}} n}{\a} \, \TSP(T). 
\end{equation}

Define a new metric $\bar{\rho}$ on $\Omega$ by $\bar{\rho}(x,y) = \rho \left( F(x),F(y) \right)$.
Since $F$ is $1$-Lipschitz, we have $\bar{\rho}(x,y) \le \abs{x-y}$. 
Note that the Wasserstein distance can only become smaller if the underlying metric 
is replaced by a smaller metric. 
Therefore, the bound \eqref{eq: W1 prelim}, which holds with respect to the usual metric $\abs{x-y}$ on $\Omega$, automatically holds with respect to the smaller metric $\bar{\rho}(x,y)$.

It remains to note that $(\Omega, \bar{\rho})$ is isometric to $(T,\rho)$. So the accuracy result \eqref{eq: W1 prelim}, which as we saw holds in $(\Omega, \bar{\rho})$, 
automatically transfers to $(T,\rho)$ (by considering the pushforward measure).
\end{proof}

\subsection{General metric spaces}

Quantization allows us to pass from discrete metric spaces to general spaces. 
A similar technique was used in Section~\ref{s: discrete to general} for the interval $[0,1]$. 
We will repeat it here for a general metric space. 

\subsubsection{Quantization}
Fix $\d > 0$ and let $\NN = \{\omega_1,\ldots,\omega_n\}$ be a $\d$-net of $T$
such that $n = \abs{\NN} = N(T,\rho,\d)$.
Consider the proximity partition 
$$
T = I_1 \cup \cdots \cup I_n
$$
where we put a point $x \in T$ into $I_i$ if $x$ is closer to $\omega_i$ that to any other points in $\NN$. (We break any ties arbitrarily.) 

We can quantize any signed measure $\nu$ on $T$ by 
defining
$$
\nu_\NN \left( \{\omega_i\} \right) = \nu(I_i), i=1,\ldots,n.
$$
Obviously, $\nu_\NN$ is a signed measure on $\NN$. Moreover, if $\nu$ is a measure, then so is $\nu_\NN$. And if $\nu$ is a probability measure, then so is $\nu_\NN$. In the latter case, it follows from the construction that 
\begin{equation}	\label{eq: W quantization general}
W_1(\nu,\nu_\NN) \le \d. 
\end{equation}
(By definition of the net, transporting any point $x$ to the closest point $\omega_i$ covers distance at most $\d$.) Furthermore, Lemma~\ref{lem: quantization contraction} easily generalizes and yields
\begin{equation}	\label{eq: TV contraction general}
\norm{\nu_\NN}_\TV \le \norm{\nu}_\TV.
\end{equation}

Finally, let us bound the TSP of the net $\NN$ using Theorem~\ref{thm: width integral}. We trivially have $N(\NN,\rho,x) \le \abs{\NN} = N(T,\rho,\d)$ for any $x>0$. Moreover, since $\NN \subset T$, we also have $N(\NN,\rho,x) \le N(T,\rho,x/2)$, see \cite[Exercise~4.2.10]{vershyninbook}. 
Using the former bound for $x<2\d$ and the latter bound for $x \ge 2\d$ and applying \eqref{eq: up to diam}, we obtain 
\begin{align} 
\TSP(\NN) 
  &\lesssim \int_0^{\diam(\NN)} N(\NN,\rho,x) \, dx \nonumber\\
  &\le 2\d N(T,\rho,\d) + \int_{2\d}^{\diam(T)} N(T,\rho,x/2) \, dx \nonumber\\
  &= 2 \left( \d N(T,\rho,\d) + \int_\d^{\diam(T)/2} N(T,\rho,x) \, dx \right) \nonumber\\
  &\leq 2 \left( 2\int_{\delta/2}^{\delta}N(T,\rho,x)\,dx + \int_\d^{\diam(T)/2} N(T,\rho,x) \, dx \right) \nonumber\\
  &\leq4\int_{\delta/2}^{\diam(T)/2}N(T,\rho,x)\,dx. \label{eq: net width}
\end{align}

\subsubsection{A private measure on a general metric space}

\begin{theorem}[Private measure on a metric space]			\label{thm: private metric}
  Let $(T,\rho)$ be a compact metric space. Let $\a, \delta >0$.
  There exists a randomized algorithm $\AA$ that takes a
  probability measure $\mu$ on $T$ as an input 
  and returns a finitely-supported probability measure $\nu$ on $T$ as an output, 
  and with the following two properties. 
  \begin{enumerate}[(i)]
    \item (Privacy): the algorithm $\AA$ is $\a$-metrically private 
      in the TV metric.
    \item (Accuracy): for any input measure $\mu$, 
    the expected accuracy of the output measure $\nu$ in the 
    Wasserstein distance is
	$$
	\E W_1(\nu,\mu) 
	\le 2\d + \frac{C}{\a} \log^{\frac{3}{2}} \left( N(T,\rho,\d) \right)\int_{\delta}^{\diam(T)}N(T,\rho,x)\,dx.
	$$
  \end{enumerate}
\end{theorem}

\begin{proof}
Preprocess the input measure $\mu$ by quantizing as in the previous subsection, 
and feed the quantized measure $\mu_\NN$ into the algorithm $\BB$ of Proposition~\ref{prop: private finite metric} for the metric space $(\NN,\rho)$. 

The contraction property \eqref{eq: TV contraction general} ensures that
$$
\norm{\mu_\NN - \mu'_\NN}_\TV \le \norm{\mu - \mu'}_\TV
$$
for any two input measures $\mu$ and $\mu'$.
This and the privacy property in Proposition~\ref{prop: private finite metric} for measures on $\NN$ guarantee
that quantization does not destroy privacy, i.e. the algorithm 
$\AA: \mu \mapsto \BB(\mu_\NN)$ is
still $\a$-metrically private as claimed.  

Next, the accuracy property in Proposition~\ref{prop: private finite metric} for the measure $\mu_\NN$ on $\NN$ gives
$$
\E W_1 \left( \BB(\mu_\NN), \mu_\NN \right)
\le \frac{C}{\a} \log^{\frac{3}{2}} (N(T,\rho,\d)) \, \TSP(\NN).
$$
Moreover, the accuracy of quantization \eqref{eq: W quantization general} states that 
$W_1(\mu,\mu_\NN) \le \d$. By triangle inequality, we conclude that 
$$
\E W_1 \left( \BB(\mu_\NN), \mu \right)
\le \d + \frac{C}{\a} \log^{\frac{3}{2}} (N(T,\rho,\d)) \, \TSP(\NN).
$$
Thus, by \eqref{eq: net width},
$$
\E W_1 \left( \BB(\mu_\NN), \mu \right)
\le \d + \frac{C}{\a} \log^{\frac{3}{2}} (N(T,\rho,\d)) \, \int_{\delta/2}^{\diam(T)/2}N(T,\rho,x)\,dx.
$$
Since $N(T,\rho,2\d)\leq N(T,\rho,\d)$, replacing $\d$ by $2\d$ completes the proof of the theorem.
\end{proof}

\subsection{Private synthetic data}

The output of the algorithm $\AA$ in Theorem~\ref{thm: private metric} is a finitely-supported probability measure $\nu$ on $T$. Quantization allows to transform $\nu$ into an {\em empirical measure} 
\begin{equation}	\label{eq: empirical measure}
\mu_Y = \frac{1}{m} \sum_{i=1}^m \d_{Y_i} 
\end{equation}
where $Y_1,\ldots,Y_m$ is some finite sequence of elements of $T$, in which repetitions are allowed. In other words, we can make the output of our algorithm a {\em synthetic data} $Y=(Y_1,\ldots,Y_m)$. Let us record this observation.

\begin{corollary}[Outputting an empirical measure]			\label{cor: empirical measure}
  Let $(T,\rho)$ be a compact metric space. Let $\a,\d>0$.
  There exists a randomized algorithm $\AA$ that takes a
  probability measure $\mu$ on $T$ as an input 
  and returns $Y = (Y_1,\ldots,Y_m) \in T^m$ for some $m$ as an output, 
  and with the following two properties. 
  \begin{enumerate}[(i)]
    \item (Privacy): the algorithm $\AA$ is $\a$-metrically private 
      in the TV metric.
    \item (Accuracy): for any input measure $\mu$, 
    the expected accuracy of the empirical measure $\mu_Y$ in the Wasserstein distance is
	$$
	\E W_1 \left( \mu_Y, \mu \right) 
	\le 3\d + \frac{C}{\a} \log^{\frac{3}{2}} \left( N(T,\rho,\d) \right) \int_{\delta}^{\diam(T)}N(T,\rho,x)\,dx.
	$$
  \end{enumerate}
\end{corollary}

\begin{proof}
Since the output probability measure $\nu$ in Theorem~\ref{thm: private metric} is finitely supported, it has the form
$$
\nu = \sum_{i=1}^r w_i \, \d_{Y_i}
$$
for some natural number $r$, positive weights $w_i$ and elements $R_i \in T$.

Let us quantize the weights $w_i$ by the uniform quantizer with step $1/m$ where $m$ is a large integer. Namely, set 
$$
q(w_i) \coloneqq \frac{\lfloor m \, w_i \rfloor}{m}.
$$
Obviously, 
the total quantization error satisfies 
\begin{equation}	\label{eq: total error}
\kappa \coloneqq \sum_{i=1}^r \left( w_i-q(w_i) \right) \in [0,r/m].
\end{equation}
To make the quantized weights a probability measure, let us add the total quantization error to any given weight, say the first. Thus, define
$$
w'_1 \coloneqq q(w_1) + \kappa
\quad \text{and} \quad 
w'_i \coloneqq q(w_i), \; i=2,\ldots,r
$$
and set
$$
\nu' \coloneqq \sum_{i=1}^r w'_i \, \d_{Y_i}.
$$

Note the three key properties of $\nu'$. 
First, since the weights $w'_i$ sum to one, $\nu'$ is a probability measure. Second, since $\nu'$ is obtained from $\nu$ by transporting a total mass of $\kappa$ across the metric space $T$, we have
$$
W_1(\nu,\nu') \le \kappa \cdot \diam(T) \le \frac{r}{m} \cdot \diam(T) \le \d
$$
where the second inequality follows from \eqref{eq: total error} and the last one by choosing $m$ large enough. Third, all quantized weights $q(w_i)$ belong to $\frac{1}{m} \Z$ by definition. Thus, $\kappa  = 1 - \sum_{i=1}^r q(w_i)$ is also in $\frac{1}{m} \Z$. Therefore, all weights $w'_i$ are in $\frac{1}{m} \Z$, too. Hence, $w'_i = m_i/m$ for some nonnegative integers $m_i$. In other words, 
$$
\nu' = \frac{1}{m} \sum_{i=1}^r m_i \, \d_{Y_i}.
$$
Since $\nu'$ is a probability measure, we must have $\sum_{i=1}^r m_i = m$. Redefine the sequence $Y_1,\ldots,Y_m$ by repeating each element $Y_i$ of the sequence $Y_1,\ldots,Y_r$ exactly $m_i$ times. Thus $\nu' = \frac{1}{m} \sum_{i=1}^m \d_{Y_i}$, as required.
\end{proof}

Corollary~\ref{cor: empirical measure} allows us to transform any true data $X=(X_1,\ldots,X_n)$ 
into a private synthetic data $Y=(Y_1,\ldots,Y_m)$.
To do this, feed the algorithm $\AA$ with the empirical measure on the true data
$\mu_X = \frac{1}{n} \sum_{i=1}^n \d_{X_i}$. 
Recall from Lemma~\ref{lem: PM to SD} that if the algorithm $\AA$ is $\a$-metrically private for $\a = \e n$, then the algorithm $X \mapsto Y = \AA(\mu_X)$ yields $\e$-differential private synthetic data. Let us record this observation:

\begin{corollary}[Differentially private synthetic data]			\label{cor: private synthetic data}
  Let $(T,\rho)$ be a compact metric space. Let $\e,\d>0$.
  There exists a randomized algorithm $\AA$
  that takes true data $X=(X_1,\ldots,X_n) \in T^n$ as an input 
  and returns synthetic data $Y = (Y_1,\ldots,Y_m) \in T^m$ for some $m$ as an output, 
  and with the following two properties. 
  \begin{enumerate}[(i)]
    \item (Privacy): the algorithm $\AA$ is $\e$-differentially private.
    \item (Accuracy): for any true data $X$, 
    the expected accuracy of the synthetic data $Y$ is
	$$
	\E W_1 \left( \mu_Y, \mu_X \right) 
	\le 3\d + \frac{C}{\e n} \log^{\frac{3}{2}} \left( N(T,\rho,\d) \right)\int_{\delta}^{\diam(T)}N(T,\rho,x)\,dx,
	$$
	where $\mu_X$ and $\mu_Y$ denote the corresponding empirical measures.  \end{enumerate}
\end{corollary}

An interested reader may now skip to Section~\ref{s: cube} where we illustrate Corollary~\ref{cor: private synthetic data} for a specific example of the metric space, namely the $d$-dimensional cube $T = [0,1]^d$. 

\bigskip

\begin{remark}[A polynomial time algorithm]  \label{re:oracle}
The algorithm in Corollary \ref{cor: empirical measure} works in {\em polynomial time with respect to the size $n$ of the input data, size $m$ of the output data, and the covering number $N(T,\rho,\d)$.} To see this, first, observe that the superregular random walk defined in (\ref{eq: superregular walk definition}) can clearly be implemented in polynomial time. Using this superregular random walk as a noise, we constructed a private signed measure in Proposition \ref{prop: signed}. Thus, the private signed measure in Proposition \ref{prop: signed} can be implemented in polynomial time. We then project this signed measure to a probability measure to obtain Proposition \ref{prop: private discrete interval}. Since this projection is done via a convex optimization, it can be implemented in polynomial time. Note that Proposition \ref{prop: private discrete interval} is in the context where the universe is a finite subset of an interval. In order to extend this to the context of general metric spaces in Theorem \ref{thm: private metric}, we first discretize $T$ by a $\d$-net of $T$ and then using a traveling salesman path, we reduce the $\d$-net of $T$ into a finite subset of an interval. The bound in Theorem \ref{thm: width integral} and its proof give a polynomial time algorithm to implement the traveling salesman path. In the case $T=[0,1]^{d}$, the traveling salesman can be done via a space-filling curve. Therefore, the algorithm in Theorem \ref{thm: private metric} can be implemented in polynomial time in $N(T,\rho,\d)$. Finally the process of turning a probability measure into a synthetic dataset in the proof of Corollary~\ref{cor: empirical measure} (quantization and replication) can clearly be implemented in polynomial time.

In the particular case where $T$ is the cube $[0,1]^d$, which we will specialize to in Corollary~\ref{cor: private synthetic cube}, our algorithm runs in time {\em polynomial in $m$ and $n$}. This is because with the optimal choice of $\d$ we make there, the covering number $N(T,\rho,\d)$ is polynomial in $n$.

Although the algorithm can be implemented in polynomial time, it is rather convoluted. We will present a streamlined and practical algorithmic implementation of a synthetic data generation method in a forthcoming paper.

\end{remark}

\section{A lower bound}\label{s:lower}

This section is devoted to impossibility results, which yield {\em lower} bounds on the accuracy of any private measure on a general metric space $(T,\rho)$. While there may be a gap between our upper and lower bounds for general metric spaces, we will see in Section~\ref{s:examples} that this gap vanishes asymptotically for spaces of Minkowski dimension $d$. 

The proof of the lower bound uses the geometric method pioneered by Hardt and Talwar \cite{hardttalwar}. A lower bound is more convenient to express in terms of packing rather than covering numbers. Recall that the {\em packing number} $\Npack(T,\rho,\e)$ of a compact metric space $(T,\rho)$ is defined 
as the largest cardinality of an $\e$-separated subset of $T$. The covering and packing numbers are equivalent up to a factor of $2$:
\begin{equation}	\label{eq: covering vs packing}
\Npack(T,\rho,2\e) \le N(T,\rho,\e) \le \Npack(T,\rho,\e),
\end{equation}
see \cite[Lemma~4.2.8]{vershyninbook}.
Thus, in all results of this section, packing numbers can be replaced by covering numbers 
at the cost of changing absolute constants.

\subsection{A master lower bound}

We first prove a general result that establishes limitations of metric privacy. 
To understand this statement better, it may be helpful to assume that 
$\MM_0=\MM_1$ and $\rho_0=\rho_1$ in the first reading.

\begin{proposition}[A master lower bound]		\label{prop: synthetic lower}
  Let $\MM_0 \subset \MM_1$ be two subsets, and  
  let $\rho_i$ be a metric on $\MM_i$, $i=0,1$.
  Assume that for some $t, \a > 0$ we have
  $$
  \diam(\MM_0, \rho_0) \le 1
  \quad \text{and} \quad
  \Npack(\MM_0, \rho_1,t) > 2e^\alpha.
  $$
  Then, for any randomized algorithm $\AA: \MM_0 \to \MM_1$ 
  that is $\alpha$-metrically private with respect to the metric $\rho_0$, 
  there exists $x \in \MM_0$ such that 
  $$
  \E \rho_1 \left( \AA(x), x \right) > t/4.
  $$
\end{proposition}

\begin{proof}
For contradiction, assume that 
\begin{equation}	\label{eq: contradiction}
\E \rho_1 \left( \AA(x), x \right) \le t/4,
\end{equation}
for all $x\in\MM_0$. Let $\NN$ be a $t$-separated subset of the metric space $(\MM_0, \rho_1)$ with cardinality 
\begin{equation}	\label{eq: net large}
\abs{\NN} > 2e^\alpha.
\end{equation} 
The separation condition implies that the balls $B(y,\rho_1, t/2)$ centered at the points $y \in \NN$ and with radii $t/2$ are all disjoint. 

Fix any reference point $y \in \MM_0$. The disjointness of the balls yields
\begin{equation}	\label{eq: sum less than 1}
\sum_{x \in \NN} \Pr{\AA(y) \in B(x,\rho_1,t/2)} \le 1.
\end{equation}
On the other hand, by the definition of $\alpha$-metric privacy,
for each $x \in \NN$ we have:
$$
\Pr{\AA(y) \in B(x,\rho_1,t/2)}
\ge \exp \left[ -\a \rho_0(x,y) \right]
	\cdot \Pr{\AA(x) \in B(x,\rho_1,t/2)}.
$$
The diameter assumption yields $\rho_0(x,y) \le 1$. 
Furthermore, using the assumption \eqref{eq: contradiction} and Markov's inequality, we obtain
$$
\Pr{\AA(x) \in B(x,\rho_1,t/2)}
= \Pr{ \rho_1 \left( \AA(x), x \right) \le t/2}
\ge \frac{1}{2}.
$$
Combining the two bounds gives 
$$
\Pr{\AA(y) \in B(x,\rho_1,t/2)} \ge \frac{1}{2e^\a}.
$$
Substitute this into \eqref{eq: sum less than 1} to get 
$$
\sum_{x \in \NN} \frac{1}{2e^\a} \le 1.
$$
In other words, we conclude that $\abs{\NN} \le 2e^\alpha$, 
which contradicts \eqref{eq: net large}. 
The proof is complete. 
\end{proof}

Can the compactness assumption of the underlying metric space $(T,\rho)$ be dropped?
 The covering number of a  general non-compact metric space of $(T,\rho)$ is infinite.  Hence, our lower bound in Proposition~\ref{prop: synthetic lower}
shows that it is impossible to establish differential privacy with any accuracy gurantees in that case without imposing additional assumptions.

\subsection{Metric entropy of the space of probability measures}

For a given compact metric space $(T,\rho)$, 
we denote by $\MM(T)$ the collection of all Borel probability measures on $T$. 
We are going to apply Proposition~\ref{prop: synthetic lower} 
for $\MM_0 = \MM_1 = \MM(T)$, 
for $\rho_1=$ Wasserstein metric and $\rho_0=$ TV metric.
That proposition requires a lower bound on the packing number 
$\Npack \left( \MM(T), W_1, t/3 \right)$. In the next lemma, we relate 
this packing number to that of $(T,\rho)$. Essentially, it says that if $T$ is large, 
then there are a lot of probability measures on $T$.

\begin{proposition}[Metric entropy of the space of probability measures]	\label{prop: entropy of space of measures}
  For any compact metric space $(T,\rho)$ and every $t>0$, we have
  $$
  \Npack \left( \MM(T), W_1, t/3 \right) 
  \ge \exp \left( c \Npack(T,\rho,t) \right),
  $$
  where $c>0$ is a universal constant\,\footnote{Throughout Sections~\ref{s:lower} and~\ref{s:examples} , $c$ will always denote the same constant, while  the constant $C$ may take on different values in different computations.}.
\end{proposition}

The proof will use the following lemma.

\begin{lemma}[A lower bound on the Wasserstein distance]	\label{lem: Wasserstein lower}
  Let $(T,\rho)$ be a $t$-separated\,\footnote{This means that the distance between any two distinct points in $T$ is larger than $t$.} compact metric space. 
  Then, for any pair of probability measures $\mu, \nu$ on $T$, we have
  $$
  W_1(\mu,\nu) \ge \mu(B^c) \, t
  \quad \text{where} \quad B = \supp(\nu).
  $$
\end{lemma}

\begin{proof}
Suppose that $\gamma$ is a coupling of $\mu$ and $\nu$. Since 
$\nu$ is supported on $B$, we have
$\gamma(B^c \times B^c) \le \gamma(T \times B^c) = \nu(B^c) = 0$, 
which means that $\gamma(B^c \times B^c)=0$. Therefore 
$$
\gamma(B^c \times B) 
= \gamma(B^c \times T) - \gamma(B^c \times B^c)
= \mu(B^c).
$$
Since the sets $B^c$ and $B$ are disjoint, the separation assumption 
implies that $\rho(x,y) > t$ for all pairs $x\in B^{c}$ and $y\in B$. Thus,
$$
\int_{T\times T}\rho(x,y)\;d\gamma(x,y)\geq\int_{B^c \times B} \rho(x,y) \; d\gamma(x,y)\geq t\gamma(B^c \times B)=t\mu(B^{c}).
$$
Since this holds for all coupling $\gamma$ of $\mu$ and $\nu$, the result follows.
\end{proof}

\begin{lemma}[Many different measures]		\label{lem: many different measures}
  Let $(\NN,\rho)$ be a $t$-separated compact metric space, and 
  assume that $\abs{\NN} \ge 2n$ for some $n \in \N$.
  Then there exists a family of at least $\exp(cn)$ empirical measures 
  on $n$ points of $T$ that are pairwise $t/3$-separated in the Wasserstein distance, where $c>0$ is a universal constant.
\end{lemma}

\begin{proof}
Let $\mu = n^{-1} \sum_{i=1}^n \d_{X_i}$ and $\nu = n^{-1} \sum_{i=1}^n \d_{Y_i}$ 
be two independent random empirical measures on $T$.
Let us condition on $\nu$ and denote $B = \supp(\nu)$. Then 
$$
\mu(B^c) = \frac{1}{n} \sum_{i=1}^n \one_{\{X_i \in B^c\}}.
$$
Now, $\one_{\{X_i \in B^c\}}$ are i.i.d.~Bernoulli random variables that take value $1$ with probability 
$$
\Pr{X_i \in B^c} 
= \frac{\abs{B^c}}{\abs{\NN}} 
\ge \frac{1}{2},
$$ 
since by construction we have $\abs{B} \le n$ and by assumption $\abs{\NN} \ge 2n$. 
Then, applying Chernoff inequality (see \cite[Exercise~2.3.2]{vershyninbook}), 
we conclude that $\mu(B^c) > 1/3$ with probability bigger than $1-e^{-5cn}$, where $c>0$ is a universal constant. 
Lemma~\ref{lem: Wasserstein lower} yields that $W_1(\mu,\nu) > t/3$.

Now consider a sequence $\mu_1,\ldots,\mu_K$ of independent random empirical measures
on $T$. Using the result above and taking a union bound we conclude that,
with probability at least $1-\binom{K}{2}e^{-5cn}$,
the inequality $W_1(\mu_i,\mu_j) > t/3$ holds for all pairs of distinct indices 
$i,j \in \{1,\ldots,K\}$. Choosing $K = \lceil e^{cn} \rceil$ makes $K$ 
between $e^{cn}$ (as claimed) and $e^{2cn}$. Thus, the success probability 
is more than $1- (e^{2cn})^2 e^{-5cn}$, which is positive. 
The existence of the required family of measures follows.
\end{proof}

\begin{proof}[Proof of Proposition~\ref{prop: entropy of space of measures}]
Let $\NN \subset T$ be a $t$-separated subset of cardinality
$\abs{\NN} = \Npack(T,\rho,t)$.
Lemma~\ref{lem: many different measures} implies the existence of 
a set of at least $\exp(c \abs{\NN})$ probability measures on $T$ 
that is $(t/3)$-separated in the Wasserstein distance. In other words, 
we have   
$\Npack \left( \MM(T), W_1, t/3 \right) \ge \exp(c \abs{\NN})$.
Proposition~\ref{prop: entropy of space of measures} is proved.
\end{proof}

\subsection{Lower bounds for private measures and synthetic data}

Now we are ready to prove the two main lower bounds on the accuracy
for (a) metrically private measures and (b) differential private data.

\begin{theorem}[Private measure: a lower bound]		\label{thm: private measure lower}
  Let $(T,\rho)$ be a compact metric space. 
  Assume that for some $t>0$ and $\a \ge 1$ we have
  $$
  \Npack(T, \rho, t) > C\alpha.
  $$
  Then, for any randomized algorithm $\AA$ that takes a
  probability measure $\mu$ on $T$ as an input 
  and returns a probability measure $\nu$ on $T$ as an output
  and that is $\alpha$-metrically private with respect to the TV metric,
  there exists $\mu$ such that 
  $$
  \E W_1(\nu,\mu) > t/12.
  $$
\end{theorem}

\begin{proof}
The assumption on the packing number for a sufficiently large constant $C$ and 
Proposition~\ref{prop: entropy of space of measures} yield
$$
\Npack \left( \MM(T), W_1, t/3 \right) \ge e^{2\a} > 2e^\a.
$$
Next, apply Proposition~\ref{prop: synthetic lower} with $t/3$ instead of $t$, 
and for $\MM_0 = \MM_1 = \MM(T)$, 
setting $\rho_1$ and $\rho_0$ to be the Wasserstein and the TV metrics, respectively. 
The required conclusion follows. 
\end{proof}

\begin{theorem}[Synthetic data: a lower bound]		\label{thm: synthetic data lower}
  There exists an absolute constant $n_0$ such that the following holds. 
  Let $(T,\rho)$ be a compact metric space. 
  Assume that for some $t>0$ and and some integer $n>n_0$ we have
  $$
  \Npack(T, \rho, t) > 2n.
  $$
  Then, for any $c$-differentially private randomized algorithm $\AA$ 
  that takes true data $X=(X_1,\ldots,X_n) \in T^n$ as an input  
  and returns synthetic data $Y=(Y_1,\ldots,Y_m) \in T^m$ for some $m$ as an output,
  there exists input data $X$ such that 
  $$
  \E W_1(\mu_Y,\mu_X) > t/12,
  $$
  where $\mu_X$ and $\mu_Y$ denote the empirical measures on $X$ and $Y$.
\end{theorem}

\begin{proof}
First note that a version of Proposition~\ref{prop: entropy of space of measures} holds for empirical measures. Namely, denote the set of all empirical measures on $n$ points of $T$ by 
$\MM_n(T)$. If $\Npack(T,\rho,t) > 2n$ then we claim that
\begin{equation}	\label{eq: N large}
\Npack \left( \MM_n(T), W_1, t/3 \right) > 2e^{c_1n}.
\end{equation}
To see this, let $\NN \subset T$ be a $t$-separated subset of cardinality
$\abs{\NN} > 2n$.
Lemma~\ref{lem: many different measures} implies the existence of 
a set of at least $e^{cn} \ge 2e^{c_1n}$ members of $\MM_n(T)$
that is $(t/3)$-separated in the Wasserstein distance. The claim \eqref{eq: N large} follows.

In preparation to apply Proposition~\ref{prop: synthetic lower}, consider the sets $\MM_0 \coloneqq T^n$ and $\MM_1 \coloneqq \cup_{k=1}^\infty T^k$. Consider
the normalized Hamming metric 
$$
\rho_0(X,X') = \frac{1}{n} \abs{\{i \in [n]:\; X_i \ne X'_i\}}
$$
on $\MM_0$, and the Wasserstein metric 
$$
\rho_1(X,X') = W_1(\mu_X,\mu_{X'})
$$
on $\MM_1$. 
Then we clearly have $\diam(\MM_0,\rho_0) \le 1$, and \eqref{eq: N large} is equivalent to
$\Npack(\MM_0,\rho_1,t/3) > 2e^{c_1 n}$.

If $\AA: \MM_0 \to \MM_1$ is a $c$-differentially private algorithm, then $\AA$ is $(cn)$-metrically private in the metric $\rho_0$ due to Lemma~\ref{lem: MP vs DP}.
Applying Proposition~\ref{prop: synthetic lower} with $t/3$ instead of $t$ and $\alpha = c_1n$, we obtain the required conclusion. 
\end{proof}

\section{Examples and asymptotics}\label{s:examples}

\subsection{A private measure on the unit cube}		\label{s: cube}

Let us work out the bound of Theorem~\ref{thm: private metric} for a concrete example: 
the $d$-dimensional unit cube equipped with the $\ell^\infty$ metric, 
i.e. $(T,\rho) = ([0,1]^d, \, \norm{\cdot}_\infty)$. The covering numbers satisfy 
$$
N(T,\norm{\cdot}_\infty,x) \le (1/x)^d, \quad x>0,
$$
since the set $x \Z^d \cap [0,1)^d$ forms an $x$-net of $T$.
Thus the accuracy is 
$$
\E W_1(\nu,\mu)
\lesssim \d + \frac{\log^{\frac{3}{2}} (1/\d)}{\a} \int_\d^1 (1/x)^d \, dx 
\lesssim \d + \frac{\log^{\frac{3}{2}} (1/\d)}{\a} \cdot (1/\d)^{d-1}
$$
if $d \ge 2$. Optimizing in $\d$ yields
$$
\E W_1(\nu,\mu) 
\lesssim \Big( \frac{\log^{\frac{3}{2}} \a}{\a} \Big)^{1/d},
$$
which wonderfully extends Theorem~\ref{thm: pm on interval} for $d=1$.
Combining the two results, for $d=1$ and $d \ge 2$, we obtain 
the following general result:

\begin{corollary}[Private measure on the cube]		\label{cor: PM cube}
  Let $d \in \mathbb{N}$ and $\a \ge 2$.
  There exists a randomized algorithm $\AA$ that takes a
  probability measure $\mu$ on $[0,1]^d$ as an input 
  and returns a finitely-supported probability measure $\nu$ on $[0,1]^d$
  as an output, and with the following two properties. 
  \begin{enumerate}[(i)]
    \item (Privacy): the algorithm $\AA$ is $\a$-metrically private in the TV metric. 
    \item (Accuracy): for any input measure $\mu$, 
    the expected accuracy of the output measure $\nu$ in the 
    Wasserstein distance is
      $$
      \E W_1(\nu,\mu) \le C \Big( \frac{\log^{\frac{3}{2}} \a}{\a} \Big)^{1/d}.
      $$
  \end{enumerate}
\end{corollary}

Similarly, by invoking Corollary~\ref{cor: private synthetic data}, we obtain 
$\e$-differential privacy for synthetic data:

\begin{corollary}[Private synthetic data in the cube]			\label{cor: private synthetic cube}
  Let $d,n \in \mathbb{N}$ and $\e>0$.
  There exists a randomized algorithm $\AA$
  that takes true data $X=(X_1,\ldots,X_n) \in ([0,1]^d)^n$ as an input 
  and returns synthetic data $Y = (Y_1,\ldots,Y_m) \in ([0,1]^d)^m$ for some $m$ as an output, 
  and with the following two properties. 
  \begin{enumerate}[(i)]
    \item (Privacy): the algorithm $\AA$ is $\e$-differentially private.
    \item (Accuracy): for any true data $X$, 
    the expected accuracy of the synthetic data $Y$ is
	$$
	\E W_1 \left( \mu_Y, \mu_X \right) 
	\le C \Big( \frac{\log^{\frac{3}{2}} (\e n)}{\e n} \Big)^{1/d},
	$$
	where $\mu_X$ and $\mu_Y$ denote the corresponding empirical measures.  \end{enumerate}
\end{corollary}

The two results above are nearly sharp in the setting when $\e$ is a constant
function of $n$ and $n \to \infty$. 
 Indeed, let us work out the lower bound 
for the cube, using Theorem~\ref{thm: private measure lower}.
The covering numbers satisfy 
$$
\Npack(T,\norm{\cdot}_\infty,x) \ge (c/x)^d, \quad x>0,
$$
which again can be seen by considering a rescaled integer grid. 
Setting $t = c/(2C\a)^{1/d}$ we get $N(T,\norm{\cdot}_\infty,t) > C\a$. Hence 
$$
\E W_1(\nu,\mu) > t/12 \gtrsim (1/\alpha)^{1/d},
$$
which matches the upper bound in Corollary~\ref{cor: PM cube} up to a logarithmic factor. 
Let us record this result. 

\begin{corollary}[Private measure on the cube: a lower bound]		\label{cor: PM cube lower}
  Let $d \in \mathbb{N}$ and $\a \ge 2$.
  Then, for any randomized algorithm $\AA$ that takes a
  probability measure $\mu$ on $[0,1]^d$ as an input 
  and returns a probability measure $\nu$ on $[0,1]^d$
  as an output, and that is $\alpha$-metrically private with respect to the TV metric,
  there exists $\mu$ such that 
  $$
  \E W_1(\nu,\mu) > c \Big(\frac{1}{\alpha} \Big)^{1/d}.
  $$
\end{corollary}

In a similar way, by invoking the lower bound in Theorem~\ref{thm: synthetic data lower},
we obtain the following nearly matching lower bound for Corollary~\ref{cor: private synthetic cube}:

\begin{corollary}[Private synthetic data in the cube: a lower bound]		\label{cor: synth data cube lower}
  Let $d,n \in \mathbb{N}$.
  Then, for any $c$-differentially private randomized algorithm $\AA$
  that takes true data $X=(X_1,\ldots,X_n) \in ([0,1]^d)^n$ as an input 
  and returns synthetic data $Y = (Y_1,\ldots,Y_m) \in ([0,1]^d)^m$ for some $m$ as an output, 
  there exists input data $X$ such that 
  $$
  \E W_1(\nu_Y,\mu_X) > c \Big(\frac{1}{n} \Big)^{1/d}.
  $$
  where $\mu_X$ and $\mu_Y$ denotes the empirical measures on $X$ and $Y$.
\end{corollary}

\begin{remark}[Low dimensions]
  As we can see, the accuracy bound $n^{-1/d}$ gets worse with increasing dimension $d$,
  and becomes constant for $d \gg \log n$. Thus, results like Corollary~\ref{cor: private    synthetic cube} are only useful for low dimensions. This should not come as a surprise. As we know from the previously mentioned no-go result by Ullman and Vadhan~\cite{ullman2011pcps}, it is computationally not feasible to construct private synthetic data in high dimensions that accurately preserves even two-way marginals, let alone all Lipschitz queries (which is what Wasserstein metric does).
\end{remark}

\subsection{Asymptotic result}

The only property of the cube $T = [0,1]^d$ we used in the previous section is the behavior on its covering numbers,\footnote{The lower bound used packing numbers, but they are equivalent to covering numbers due to \eqref{eq: covering vs packing}.} namely that 
\begin{equation}	\label{eq: covering asymptotic}
N(T,\rho,x) \asymp (1/x)^{-d}, \quad x>0.
\end{equation}
Therefore, the same results on private measures and synthetic data 
hold for any compact metric space $(T,\rho)$ whose covering numbers behave this way. 
In particular, it follows that any probability measure $\mu$ on $T$ can be transformed into a $\alpha$-metrically private measure $\nu$ on $T$, with accuracy
\begin{equation}	\label{eq: private measure asymptotics}
\E W(\nu, \mu) \asymp (1/\a)^{1/d}.
\end{equation}
(ignoring logarithmic factors), and this result is nearly sharp. Similarly, any true data $X \in T^n$ can be transformed into $\e$-differentially private synthetic data $Y \in T^m$ for some $m$, with accuracy 
\begin{equation}	\label{eq: synthetic data asymptotics}
\E W(\mu_Y, \mu_X) \asymp (1/n)^{1/d}.
\end{equation}
(ignoring logarithmic factors and dependence on $\e$), and this result is nearly sharp. 

These intuitive observations can be formalized using the notion of {\em Minkowski dimension}. By definition, the metric space $(T,\rho)$ has Minkowski dimension $d$ if 
$$
\lim_{x \to 0} \frac{\log N(T,\rho,x)}{\log(1/x)} = d.
$$
The following two asymptotic results combine upper and lower bounds, and 
essentially show that \eqref{eq: private measure asymptotics} and \eqref{eq: synthetic data asymptotics} hold in any space of dimension $d$.

\begin{theorem}[Private measure, asymptotically]		\label{thm: asymp PM}
  Let $(T,\rho)$ be a compact metric space of Minkowski dimension $d\ge1$. 
  Then
  $$
  \lim_{\a \to \infty} \inf_\AA \sup_\mu \frac{\log ( \E W_1(\AA(\mu), \mu) )}{\log \a} = -\frac{1}{d}.
  $$
  Here the infimum is over randomized algorithms $\AA$ that input and output a probability measure on $T$ and are $\alpha$-metrically private with respect to the TV metric;
  the supremum is over all probability measures $\mu$ on $T$.
\end{theorem}

\begin{proof}
We deduce the upper bound from Theorem~\ref{thm: private metric} and the lower bound from 
Theorem~\ref{thm: private measure lower}.

\medskip

{\em Upper bound.} 
By rescaling, we can assume without loss of generality that $\diam(T,\rho) = 1$.
Fix any $\e >0$. By definition of Minkowski dimension, there exists $\d_0>0$
such that 
\begin{equation}	\label{eq: Minkowski unlimit}
N(T,\rho,x) \le (1/x)^{d+\e}
\quad \text{for all } x \in (0,\d_0).
\end{equation}
Then
$$
\int_\d^1 N(T,\rho,x) \, dx 
\le \int_\d^{\d_0} (1/x)^{d+\e} \, dx + \int_{\d_0}^1 N(T,\rho,x) \, dx
\le K(1/\d)^{d+\e-1} + I(\d_0)
$$
where $K = 1/(d+\e-1)$ and $I(\d_0) = \int_{\d_0}^1 N(T,\rho,x) \, dx$.
The last step follows if we replace $\d_0$ by infinity and compute the integral.

If we let $\d \downarrow 0$, we see that $K(1/\d)^{d+\e-1} \to \infty$ while $I(\d_0)$ stays the same since it does not depend on $\d$. Therefore, there exists $\d_1>0$ such that $I(\d_0) \le K(1/\d)^{d+\e-1}$ for all $\d \in (0,\d_1)$. 
Therefore,
$$
\int_\d^1 N(T,\rho,x) \, dx \le 2K(1/\d)^{d+\e-1}
\quad \text{for all } \d \in (0,\min(\d_0,\d_1)).
$$
Applying Theorem~\ref{thm: private metric} for such $\d$ and using \eqref{eq: Minkowski unlimit}, we get 
\begin{equation}	\label{eq: balanced terms}
\inf_\AA \sup_\mu \E W_1(\nu,\mu) 
\le 2\d + \frac{C}{\a} \log^{\frac{3}{2}} \left( (1/\d)^{d+\e} \right) \cdot 2K(1/\d)^{d+\e-1}.
\end{equation}
Optimizing in $\d$, we find that a good choice is 
$$
\d = \d(\a) = \left( \frac{\log^{\frac{3}{2}}(K\a)}{K\a} \right)^{\frac{1}{d+\e}}.
$$
For any sufficiently large $\a$, we have $\d < \min(\d_0,\d_1)$ as required, 
and substituting $\d = \d(\a)$ into the bound in \eqref{eq: balanced terms} we get after simplification:
$$
\inf_\AA \sup_\mu \E W_1(\nu,\mu) 
\le (1+2CK)\d(\a).
$$

Furthermore, recalling that $K$ does not depend on $\alpha$, it is clear that 
$$
\lim_{\a \to \infty} \frac{\log \left( (1+2CK)\d(\a) \right)}{\log \a} 
= -\frac{1}{d+\e}.
$$
Thus
$$
\limsup_{\a \to \infty} \frac{\log( \inf_\AA \sup_\mu \E W_1(\nu,\mu) )}{\log \a} 
\le -\frac{1}{d+\e}.
$$
Since $\e>0$ is arbitrary, it follows that 
\begin{equation}	\label{eq: upper limit}
\limsup_{\a \to \infty} \inf_\AA \sup_\mu \frac{\log ( \E W_1(\AA(\mu), \mu) )}{\log \a} 
\le -\frac{1}{d}.
\end{equation}

\medskip

{\em Lower bound.} 
Fix any $\e>0$. By definition of Minkowski dimension and the equivalence \eqref{eq: covering vs packing}, there exists $\d_0>0$
such that 
$$
\Npack(T,\rho,x) \ge N(T,\rho,x) > (1/x)^{d-\e}
\quad \text{for all } x \in (0,\d_0).
$$
Set
$$
x(\a) = \left( \frac{1}{C\a} \right)^{\frac{1}{d-\e}}.
$$
Then, for any sufficiently large $\a$, we have $x \in (0,\d_0)$ and 
$$
\Npack(T,\rho,x(\a)) > C\a.
$$
Applying Theorem~\ref{thm: private measure lower}, we get 
$$
\inf_\AA \sup_\mu \E W_1(\nu,\mu) 
\ge x(\a)/20.
$$
It is easy to check that 
$$
\lim_{\a \to \infty} \frac{\log \left( x(\a)/20 \right)}{\log \a}
= -\frac{1}{d-\e}.
$$
Thus
$$
\liminf_{\a \to \infty} \frac{\log( \inf_\AA \sup_\mu \E W_1(\nu,\mu) )}{\log \a} 
\ge -\frac{1}{d-\e}.
$$
Since $\e>0$ is arbitrary, it follows that 
$$
\liminf_{\a \to \infty} \inf_\AA \sup_\mu \frac{\log ( \E W_1(\AA(\mu), \mu) )}{\log \a} 
\ge -\frac{1}{d}.
$$
Combining with the upper bound \eqref{eq: upper limit}, we complete the proof.
\end{proof}

In a similar way, we can deduce the following asymptotic result for private synthetic data. 
The argument is analogous; the upper bound follows from Corollary~\ref{cor: private synthetic data} and the lower bound from Theorem~\ref{thm: synthetic data lower}.

\begin{theorem}		\label{thm: asymp synth data}
  Let $(T,\rho)$ be a compact metric space of Minkowski dimension $d\ge1$. 
  Then, for every $\e \in (0,c)$, we have
  \begin{equation}\label{minkovskibound}
  \lim_{n \to \infty} \inf_\AA \sup_X \frac{\log ( \E W_1(\mu_Y, \mu_X) )}{\log n} = -\frac{1}{d}.
  \end{equation}
  Here the infimum is over $\e$-differentially private randomized algorithms $\AA$ 
  that take true data $X=(X_1,\ldots,X_n) \in T^n$ as an input 
  and return synthetic data $Y = \AA(X) = (Y_1,\ldots,Y_m) \in T^m$ for some $m$ as an output.
\end{theorem}

\begin{remark}[Low-dimensional data in high dimensions?]
  This and other results proved here show that
  the accuracy of private synthetic data must deteriorate quickly 
  as the data dimension $d$ increases. 
  But does this mean that the proposed method is useless for any high-dimensional data? 
  In practice, this is not necessarily the case. Real-world high-dimensional data often live in (or near) a low-dimensional smooth manifold. Since a smooth manifold is metrizable and a smooth $d$-dimensional manifold in $\R^n$ has Minkowski dimension $d$, our framework may still apply to the standard setting of high-dimensional statistics, where the data lives in high dimension but its intrinsic geometry is low-dimensional.
Thus, the generalization via Minkowski dimension presented in this section may not only be appealing from a theoretical viewpoint, but carries a practical potential, which we hope to pursue in our future work.
\end{remark}

\section*{Acknowledgement}

M.B.\ acknowledges support from NSF DMS-2140592. T.S.\ acknowledges support from  NSF DMS-2027248, NSF CCF-1934568, NIH R01HL16351, and a CeDAR Seed grant.
 R.V.\ acknowledges support from NSF DMS-1954233, NSF DMS-2027299, U.S.~Army 76649-CS, and NSF+Simons Research Collaborations on the Mathematical and Scientific Foundations of Deep Learning.


\end{document}